\documentclass[12pt, draftclsnofoot,]{IEEEtran}

\usepackage{amssymb}
\usepackage{amsfonts}
\usepackage{epsf}
\usepackage{epsfig}
\usepackage{epstopdf}
\usepackage{array}
\usepackage{amsmath}
\usepackage{url}
\usepackage[algoruled, linesnumbered]{algorithm2e}
\usepackage{amsthm}
\usepackage{xifthen}
\usepackage{mathrsfs}
\usepackage[mathscr]{euscript}
\usepackage[sans]{dsfont}
\usepackage[dvipsnames]{xcolor}
\usepackage{graphicx}
\usepackage{caption}
\usepackage{subcaption}

\usepackage{setspace}	

\usepackage{tikz}
\usetikzlibrary{automata,arrows,positioning,calc}

\usepackage{epsf}
\usepackage{epsfig}

\usepackage{array}

\usepackage{times}
\usepackage{graphicx}
\usepackage{amsmath}
\usepackage{amssymb}
\usepackage{amsxtra}
\usepackage{here}
\usepackage{rawfonts}
\usepackage{times}
\usepackage{url}
\usepackage{cite}
\usepackage{amsthm}
\usepackage{graphicx}
\usepackage{caption}

\usepackage{mathtools}
\usepackage{url}
\usepackage[algoruled,linesnumbered]{algorithm2e}

\usepackage{cite}

\usepackage{cite} 

\usepackage[T1]{fontenc}
\usepackage[utf8]{inputenc}
\usepackage{authblk}

\newtheorem{thm}{Theorem}
\newtheorem{lem}[thm]{Lemma}
\newtheorem{prop}{Proposition}

\newtheorem{remark}{Remark}

\theoremstyle{definition}

\ifCLASSINFOpdf
\else
\fi

\onecolumn

\begin{document}

%


\title{\LARGE   \hspace{50 cm}A Multiclass Mean-Field Game for Thwarting Misinformation Spread in the Internet of Battlefield Things (IoBT) \vspace{-0.5 ex} }

\author[D. Kopta et al.]
{\normalsize Nof Abuzainab and Walid Saad
	\\ \vspace{-0.4 cm}
	Wireless@VT, Department of Electrical and Computer Engineering, Virginia Tech, Blacksburg, VA, USA, Emails:\{nof, walids\}@vt.edu\\ 
	\vspace{-5ex}
\thanks{This research was sponsored by the Army Research Laboratory and was
accomplished under Grant Number W911NF-17-1-0021. }
\thanks{A preliminary version of this work \cite{confc}  was submitted for conference publication.}
}

\maketitle


%
\IEEEpeerreviewmaketitle
\vspace{-0.8 cm}
\begin{abstract}
In this paper, the problem of misinformation propagation is studied for an Internet of Battlefield Things (IoBT) system in 
which an attacker seeks to inject false information in the IoBT nodes in order to compromise the IoBT operations. In the considered model, each IoBT node seeks to counter the misinformation attack by finding the optimal probability of accepting a given information that minimizes its cost at each time instant. The cost is expressed in terms of the quality of information received as well as the infection cost. The problem is formulated as a mean-field game with multiclass agents which is suitable to model a massive heterogeneous IoBT system. For this game, the mean-field equilibrium is characterized, and an algorithm based on the forward backward sweep method is proposed to find the mean-field equilibrium. Then, the finite IoBT case is considered, and the conditions of convergence of the equilibria in the finite case to the mean-field equilibrium are presented. Numerical results show that the proposed scheme can achieve a 1.2-fold increase in the quality of information (QoI) compared to a baseline scheme in which the IoBT nodes are always transmitting. The results also show that the proposed scheme can reduce the proportion of infected nodes  by $99\%$ compared to the baseline.

\end{abstract}
\section{Introduction}

With the advent of the Internet of Things (IoT), next-generation military networks will rely more on machine intelligence and the information collected from the densely deployed IoT devices \cite{IoBT,IoBTtwo,IoBTthree}. The integration of military networks with the various IoT devices will potentially achieve battlefield autonomy and increase considerably the efficiency of battlefield operations, thus forming the so-called the \emph{Internet of battlefield Things (IoBT)}\cite{IoBTone}. However, due to its adversarial nature, the IoBT is prone to a multitude of security attacks. One important attack on IoBT is the misinformation attack \cite{IoBT} in which an adversary injects false information at each IoBT device. Such misinformation can then be used by the adversary to manipulate the decisions of the military commanders, in an effort to jeopardize the success of the military mission.  Thus, realizing the vision of a large-scale IoBT is largely contingent on developing novel security mechanisms to combat misinformation propagation across the various IoBT nodes.

The dynamics of misinformation propagation have been recently modeled using epidemic models for social networks in \cite{social} and for mobile  opportunistic networks such as those encountered in the battlefield in \cite{mobility}.
Epidemic models are suitable for IoBT misinformation propagation due to the presence of strong interactions among the densely deployed IoBT devices. This dense nature of the IoBT implies that an IoBT device can get easily infected with misinformation whenever it communicates with any one of its infected neighbors. Further, epidemic models can capture systems with infinite number of nodes, which is suitable for the large-scale IoBT systems.

Many existing works have considered the problem of controlling the spread of network epidemics and studied the interaction between the network and the adversary using game-theoretic approaches \cite{epidemicgame2, epidemicgame3, epidemicgame7, epidemicgame1, epidemicgame0, epidemicgame4, epidemicgame5,epidemicgame6, epidemicgame9, mean-fieldSIR}.
In \cite{epidemicgame2, epidemicgame3, epidemicgame7}, a noncooperative game is considered in which the players are the network nodes whose goal is to choose a curing rate that minimizes the protection cost as well as the infection costs at steady-state.  In \cite{epidemicgame1}, several noncooperative games are proposed for network epidemic control between a network operator and an attacker with the goal of minimizing the infection cost. A zero-sum differential game  is proposed  in  \cite{epidemicgame0} and \cite{epidemicgame4} for network malware propagation in which the network operator controls the recovery rates of the sensor nodes whereas an attacker chooses the infection rate that maximizes the infection cost. The work in \cite{epidemicgame0} particularly considers a wireless sensor network in which the network operator controls the sleep rate of the sensor nodes in addition to the recovery rate in order to limit the spread of infections. 
The authors in \cite{epidemicgame5} propose a network formation game  in which the network nodes choose to construct links starting from an empty network in order to reach a connected, steady-state network while minimizing the costs of infection. 
 In \cite{epidemicgame6} and \cite{epidemicgame9}, the problem of controlling the network epidemic through vaccination is formulated as a zero-determinant game where both the network administrator and the nodes are the players. In \cite{mean-fieldSIR}, a mean-field game is proposed to study infection spread in a fully connected regular network.

However, most of this prior art \cite{epidemicgame2, epidemicgame3, epidemicgame7, epidemicgame1, epidemicgame0, epidemicgame4, epidemicgame5,epidemicgame6, epidemicgame9, mean-fieldSIR} models the network as either a fully connected graph or as a $k$-regular graph. However, in an IoBT, the nodes have heterogeneous connectivty, and, thus, there is a need to consider more suitable graph models that account for the IoBT heterogeneity.  Also, considering the network operator as the sole network player as done in \cite{epidemicgame1, epidemicgame0, epidemicgame4} is not suitable for the IoBT since it requires a centralized control over all of the IoBT nodes and therefore significant time and control overheads, which is not tolerable in time-sensitive military missions. Hence, distributed approaches are more favorable for the IoBT since the nodes must instanteneously take control to limit  misinformation propagation. Moreover, existing works, such as in\cite{epidemicgame2} and \cite{epidemicgame3}, that consider the network nodes as the players typically seek to maximize the payoff when the system is at the steady state. Such approaches are not suitable for the problem of misinformation propagation in IoBT. This is due to the fact that information propagation in the IoBT is time sensitive. Thus, in order to maintain the successful operation of the IoBT, it is critical to limit the spread of misinformation at  each time instant and not only at the steady state. In addition, choosing only the curing rate, as done by most of the existing works \cite{epidemicgame2, epidemicgame3, epidemicgame7, epidemicgame1, epidemicgame0, epidemicgame4, epidemicgame5,epidemicgame6, epidemicgame9, mean-fieldSIR}, is not adequate to instanteneously limit the spread of misinformation. In fact, curing the nodes comes at the expense of security costs. Thus, there is a need to implement cost-efficient actions that can effectively limit the spread of misinformation. Here, it is also worth noting that, recently, a number of works \cite{MIoT1, MIoT2, MIoT3, MIoT4} have studied various IoBT security scenarios, however, these works do not analyze the critical problem of misinformation spread. 

The main contribution of this paper is a novel, comprehensive framework for thwarting the spread of misinformation in a large-scale IoBT system. In particular, the proposed framework will yield the following key contributions:
\begin{itemize}
\item We propose a novel approach to control misinformation propagation in the IoBT. In particular, we propose a distributed approach in which each IoBT node decides whether or not to accept the received information at each time instant, in order to to limit the propagation of misinformation. Thus, our proposed approach, due to its distributed nature, is scalable for a large-scale system such as the IoBT. Further, due to the heterogeneity of the IoBT nodes in terms of connectivity, we model the IoBT as a random graph in which the nodes have heterogeneous degrees that follow a predetermined distribution.
\item We consider an epidemic model suitable for IoBT misinformation propagation that accounts for the heterogeneous characteristics of the IoBT nodes to effectively identify misinformation. 
In particular, we consider an SELI epidemic model for the IoBT in which, in addition to the conventional susceptible ($S$) and infected ($I$) states of the nodes, we introduce latent ($L$) and exposed ($E$) states to capture scenarios in which the IoBT nodes choose to perform further processing to check the validity of any received information.
\item We formulate the IoBT misinformation propagation problem as a finite-state mean-field game \cite{finitemean-field} with multiclass agents  \cite{multiclass} whose players are the IoBT nodes each of which is seeking to determine its probability of accepting the received information. Mean-field games \cite{meanw} are suitable for our problem since they handle an infinite number of players which is the case for a large scale IoBT. 
Further, the proposed framework of mean-field games with \emph{multiclass agents} can capture the presence of several types of populations where agents belonging to the same type have similar characteristics.
Thus, such games \cite{multiclass} are suitable to model the heterogeneous characteristics of the IoBT nodes, unlike conventional mean-field games \cite{mean-fieldSIR} that assume all players to be similar. To the best of our knowledge, there is no prior work that combines mean-field games with multiclass agents, as proposed here.
\item We consider a suitable metric for misinformation propagation known as the quality-of-information (QoI) in the  IoBT nodes' payoff in addition to the infection cost.  The QoI of each node is defined as a function of the information received from its neighbors, the integrity of the received information and the age of information.

\item We extend the definition of the finite state mean-field equilibrium (MFE) that is defined in \cite{finitemean-field} to the case of multiclass agents and propose an algorithm based on the forward backward sweep method \cite{FBSM} to find the MFE. Then, we analyze the case of a finite IoBT game and prove the convergence of the  equilibrium of the finite game to the MFE. This result, in turn, shows that the MFE is an effective approximation to a real-world IoBT system with a large number of players.
\item Numerical results show that the proposed scheme can achieve a 1.2-fold increase in the quality of information (QoI) compared to a baseline scheme in which the IoBT nodes are always transmitting. Further, the proposed scheme can reduce the proportion of infected nodes  by $99\%$ compared to the baseline.
\end{itemize}
\vspace{-1 cm}
The rest of the paper is organized as follows: Section II presents the system model. Section III present the IoBT mean-field game. Section IV presents the finite IoBT game and the convergence conditions of the finite game to the mean-field game. Section V presents the simulation results. Finally, conclusions are drawn in section VI.

\section{System Model}
\vspace{-0.1 cm}
\subsection{Epidemic  Model}
Consider an IoBT network modeled by a random graph whose node degrees are distributed according to a distribution $P$, with $P(k)$ being the probability that an IoBT node has a degree $k$. We let $K_{\max}$ be the maximum degree in the IoBT system. The random graph is a realistic model for a large-scale IoBT due to the heterogeneity in the connectivity of the IoBT nodes. In fact, the IoBT network comprises IoBT devices that are locally connected to a cluster head, and sinks/fusion centers that are connected to a multitude of cluster heads and various IoBT devices.  Further, by properly choosing the degree distribution $P$, the random graph can represent a tree-like structure which commonly represents the topology of military networks \cite{treelike}. In the considered IoBT graph, nodes having degree $k$ are classified into types within a set $\mathcal{H}_k$ and distributed according to $p_k$, where $p_k(i)$ is the probability that an node of degree $k$ has type $i \in \mathcal{H}_k$. The existence of multiple types of nodes having a degree $k$ stems from the heterogeneity of the IoBT nodes that include simple sensors, wearables, vehicles, cameras, and robots or drones that have different capabilities, characteristics, and roles.


In this IoBT, an attacker seeks to inject false information into the nodes in order to disrupt the normal operations of the system. Let $\lambda_{ik}$ be the rate of injection of false information into a node of degree $k$ and type $i$. At any given time instant $t$, each IoBT node chooses to either accept the received information and then transmit it or to doubt the integrity of this received information. An IoBT node becomes infected once it accepts false information. When the IoBT node chooses to doubt the information, it retains the information for some time for further inspection. For instance, it can potentially run a classification machine learning algorithm to decide whether to forward or discard the stored information. Finally, an infected node no longer uses the misinformation when it becomes obsolete and, hence, this node goes back to being susceptible to attacks.
Thus, each one of the IoBT nodes can be in one of the following states:
\begin{itemize}
\item \emph{Susceptible (S)}: A node is said to be susceptible when it does not contain misinformation but, simultaneously, it does not have strong security mechanisms to identify misinformation. Hence, it can get infected with misinformation either when it accepts information forwarded from an infected node or when the attacker succeeded in injecting  misinformation directly into the designated susceptible node.
\item \emph{Exposed (E)}: A node is said to be exposed when it receives misinformation from its neighbour. Yet, it is doubtful about the credibility of the received information, and, hence, it does not immediately forward the information to its neighbour.
\item \emph{Latent (L)}: A node is said to be latent when it receives true information. However, it still decides to do further processing to inspect the received information.
\item \emph{Infected (I)}: A node is said to be infected when it contains misinformation which it believes it is correct, and, subsequently, it forwards the misinformation to its neighbour.
\end{itemize}

\vspace{-0.3 cm}
The SEI model and its variants have been been commonly adopted as a realistic model to analyze misinformation propagation in networks (e.g. see \cite{social, mobility}). However, existing models do not consider the case when the information which a node doubts is true, which is different from the case when the node doubts misinformation. This is due to the fact that the probability of accepting information after processing generally depends on whether the information is true or not.
Thus, we introduce the latent (\emph{L}) state to represent tin which the IoBT nodes analyze true information. Further, existing models such as those in \cite{mobility} do not explicitly take into account the delay incurred when a node decides to analyze its received information. In contrast, in our model, we account for the processing delay in the latent and the exposed states by the probabilities of residing in states $E$ and $L$, respectively.
\begin{figure}
\begin{center}
\begin{tikzpicture}[->, >=stealth', auto, semithick, node distance=4.5 cm]
\tikzstyle{every state}=[fill=white,draw=black,thick,text=black,scale=1]
\node[state]    (A)                     {$S$};
\node[state]    (B)[above right of=A]   {$E$};
\node[state]    (C)[below right of=A]   {$L$};
\node[state]    (D)[below right of=B]   {$I$};
\path
(A) 
    edge[bend left]     node{$\sigma_{ik}(t)R_{ik}(\Theta(t))$}     (B)
    edge [bend left, above]    node{$\hspace{0.4 cm}\alpha_{ik}(t)R_{ik}(\Theta(t))$}      (D)
    edge   node{$\sigma_{ik}(t)L_{ik}(\Theta(t))$}      (C)
(B) edge[bend left]               node{$\beta^E_{ik}(1-\delta_{ik})$}           (D)
      edge [above]              node{$\hspace{2.8 cm} \gamma^E_{ik}(1-\delta_{ik})$}           (A)
(C) 
      edge[bend left] node{$(\beta^L_{ik}+\gamma^L_{ik})(1-\delta_{ik})$} (A)
(D) 
    edge   node{$\nu_{ik}$}         (A)    
;
\end{tikzpicture}
\end{center}
 \caption{State transition diagram of IoBT node of degree $k$ and type $i$}\label{StateTrans}
\vspace{-0.2 cm}
\end{figure}
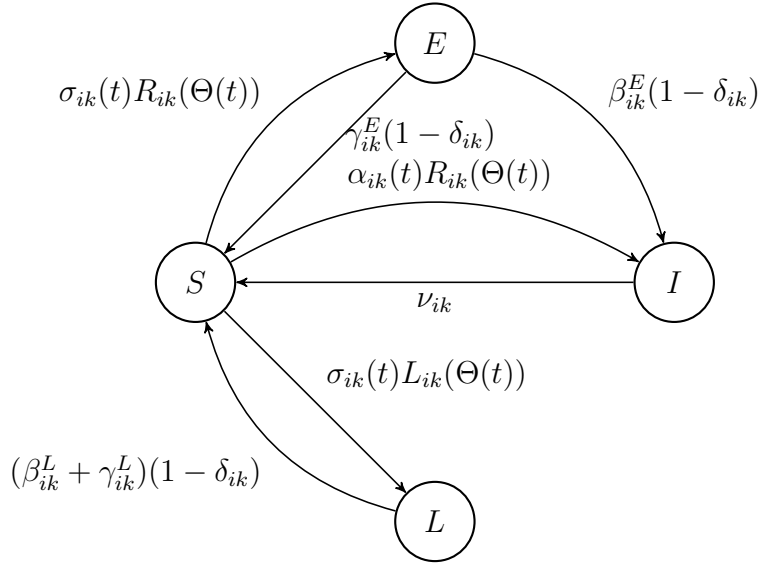
Fig. \ref{StateTrans} shows the state transition diagram of each IoBT node of degree $k$ and type $i$. 
When an IoBT node of degree $k$ and type $i$ is susceptible and receives information, it will either accept the information with probability $\alpha_{ik}(t)$ and become infected or it will doubt the information with probability $\sigma_{ik}(t)=1-\alpha_{ik}(t)$. The IoBT node receives misinformation either from infected neighbours or directly from the attacker. Based on \cite{hetero}, the probability with which a node with degree $k$ is infected by one of its neighbours is derived as $k\Theta(t)$ where $\Theta(t)$ is the probability that a randomly chosen link is pointing to an infected node and is given by

\begin{equation}
\Theta(t)=\frac{\sum_k kP(k)\sum_{i \in \mathcal{H}_k}p_k(i) I_{ik}(t)}{<k>}, \label{theta}
\end{equation}
where $<k>=\sum_{k}kP(k)$ and $I_{ik}(t)$ is the proportion of infected devices of degree $k$ and type $i$. Thus, since the attacker directly infects the node with rate $\lambda_{ik}$, the total infection rate is $\alpha_{ik}(t)R_{ik}(\Theta(t))$ where $R_{ik}(\Theta(t))=\lambda_{ik}+k\Theta(t)$. 
Similarly, when the IoBT node is susceptible and receives true information,  it becomes latent with probability $\sigma_{ik}(t)L_{ik}(\Theta(t))$ or remains susceptible with probability $\alpha_{ik}(t)L_{ik}(\Theta(t))$, where $L_{ik}(\Theta(t))=(1-\lambda_{ik})(1-\Theta(t))^k$ is the probability that an IoBT node does not receive misinformation at time $t$.

 When an IoBT node is in the exposed state, it remains in this state with probability $\delta_{ik}$. The probability $\delta_{ik}$ corresponds to the delay incurred to judge the credibility of the obtained information. Then, the node will accept the information with probability $\beta^E_{ik}$ and become infected, or it will refuse the misinformation and return back to the susceptible state with probability $\gamma^E_{ik}$.
Similarly, when the node is in the latent state, it remains in this state in order to process the information with probability $\delta_{ik}$,
 it accepts the information with probability $\beta^L_{ik}$ or rejects it with probability $\gamma^L_{ik}$. When in the latent state, the IoBT node will return to the susceptible state whether it decides to accept or reject the information (i.e. with probability $\beta^L_{ik}+\gamma^L_{ik}$).
The probabilities  $\delta_{ik}$, $\beta^E_{ik}$, $\gamma^E_{ik}$, $\beta^L_{ik}$ and $\gamma^L_{ik}$   depend on the node's capabilities (for example the strength of the machine learning algorithm used) and its effectiveness in identifying the misinformation.
Finally, an IoBT node of degree $k$ and type $i$ discards the misinformation when it is no longer useful after some time with probability $\nu_{ik}$.

 Let $m^S_{ik}(t)$, $m^{E}_{ik}(t)$, $m^{L}_{ik}(t)$, and $m^{I}_{ik}(t)$ be the proportions of IoBT nodes of degree $k$ and type $i$ in states $\emph{S}$, $\emph{E}$,  $\emph{L}$, and $\emph{I}$, respectively. Since IoBT networks typically have a massive number of devices, we consider a mean-field epidemic model. Mean-field models provide a simple yet effective representation of large and complex interacting system of agents.  Mean-field models mainly study decision making processes in which the number of agents tends to infinity, and that the dynamics of agents of similar characteristics can be described by an aggregate behavior. Further, in mean-field models since the number of agents tend to infinity, the influence of a single agent on the overall system is negligible, yet the effect of all agents is considerable and approximated by their average effect.
Thus, at the mean-field level, the state dynamics of IoBT nodes of degree $k$ and type $i$ are governed by the following Kolmogorov differential equations:
\vspace{-0.6 cm}

\small
\begin{eqnarray}
&&\hspace{-1.3 cm}\frac{\partial m^{S}_{ik}(t)}{\partial t}=\hspace{-0.1 cm}-(1-\bar{\alpha}_{ik}(t)+\bar{\alpha}_{ik}(t) R_{ik}(\Theta(t))m^{S}_{ik}(t)+(1-\delta_{ik})\gamma^E_{ik}m^E_{ik}+(1-\delta_{ik})m^L_{ik}+\nu_{ik}m^I_{ik},\label{ms}
\end{eqnarray}
\begin{eqnarray}
&&\hspace{-7 cm}\frac{\partial m^{E}_{ik}(t)}{\partial t}=-(1-\delta_{ik})m^{E}_{ik}+(1-\bar{\alpha}_{ik}(t))R_{ik}(\Theta(t))m^{S}_{ik}(t),\label{me}
\end{eqnarray}
\vspace{-0.6 cm}
\begin{eqnarray}
&&\hspace{-7 cm}\frac{\partial m^{L}_{ik}(t)}{\partial t}=-(1-\delta_{ik})m^{L}_{ik}+(1-\bar{\alpha}_{ik}(t))L_{ik}(\Theta(t))m^{S}_{ik}(t),\label{ml}
\end{eqnarray}
\begin{eqnarray}
&&\hspace{-5.8 cm}\frac{\partial m^{I}_{ik}(t)}{\partial t}= \bar{\alpha}_{ik}(t) R_{ik}(\Theta(t))m^{S}_{ik}(t)+(1-\delta_{ik})\beta^E_{ik}m^E_{ik}-\nu_{ik}m^I_{ik}(t),
\label{mI}
\end{eqnarray}
\normalsize
where $\bar{\alpha}^{S}_{ik}(t)$ is the aggregate rate of accepting information for all nodes with degree $k$ and type $i$ when in the susceptible state. At time $0$, all the nodes are susceptible, and, thus, $m^{S}_{ik}(0)=1$ and $m^I_{ik}(0)=m^{E}_{ik}(0)=m^{L}_{ik}(0)=0, \hspace{0.2 cm} \forall i,k.$

Each IoBT node of degree $k$ and type $i$ seeks to determine the probability $\alpha_{ik}(t)$ of accepting the received information that minimizes its cost. Let $x^{S}_{ik}(t)$, $x^{E}_{ik}(t)$, $x^{L}_{ik}(t)$, and $x^I_{ik}(t)$ be the probabilities that the node of degree $k$ and type $i$ is in the $S$, $E$, $L$, and $I$ states, respectively.  Based on the epidemic dynamics in (\ref{mI}), the state $\boldsymbol{x}_{ik}(t)=(x^{S}_{ik}(t),x^{E}_{ik}(t),x^L_{ik}(t), x^I_{ik}(t))$ of each node is governed by the following differential equations:
\vspace{-0.4 cm}

\small
\begin{eqnarray}
&&\hspace{-2cm}\frac{\partial x^{S}_{ik}(t)}{\partial t}=\hspace{-0.1 cm}-(1-\bar{\alpha}_{ik}(t)+\bar{\alpha}_{ik}(t) R_{ik}(\Theta(t))x^{S}_{ik}(t)+(1-\delta_{ik})\gamma^E_{ik}x^E_{ik}+(1-\delta_{ik})m^L_{ik}+\nu_{ik}x^I_{ik},\label{xs}\\
&&\hspace{-2 cm}\frac{\partial x^{E}_{ik}(t)}{\partial t}=-(1-\delta_{ik})x^{E}_{ik}+(1-\bar{\alpha}_{ik}(t))R_{ik}(\Theta(t))x^{S}_{ik}(t),\label{xE}\\
&&\hspace{-2 cm}\frac{\partial x^{L}_{ik}(t)}{\partial t}=-(1-\delta_{ik})x^{L}_{ik}+(1-\bar{\alpha}_{ik}(t))L_{ik}(\Theta(t))x^{S}_{ik}(t),\label{xL}\\
&&\hspace{-2 cm}\frac{\partial x^{I}_{ik}(t)}{\partial t}= \bar{\alpha}_{ik}(t) R_{ik}(\Theta(t))x^{S}_{ik}(t)+(1-\delta_{ik})\beta^E_{ik}x^E_{ik}-\nu_{ik}x^I_{ik}(t).\label{xI}
\end{eqnarray}

\normalsize
\vspace{-0.5 cm}
\subsection{Cost Functions}
The payoff of each IoBT node is expressed in terms of the \emph{quality of its information} and the cost of infection. The QoI is a metric that has been widely used to assess the information generated by sensors networks, in general, and military networks as discussed in particular \cite{qqq} and \cite{QoI2}. It is mainly a function of the precision of the sensing device, the integrity, and the age of information. Further, in \cite{QoC}, the authors consider the QoI of the information received by each node as an increasing function of the number of its neighbors. However, the integrity of the information is not considered. 
In contrast, in our IoBT problem, we define the QoI $Q_{ik}(t)$  to be a joint function of the degree of the node, the information integrity, and the age of information.
Consequently, the QoI of a given IoBT node of degree $k$ and type $i$ will be given by
\begin{equation}
Q_{ik}(t)=V_{ik}(t)-\kappa\delta_{ik},
\end{equation}
where $\delta_{ik}$ is the delay of processing the information when in the latent or exposed state, $\kappa$ is a normalization constant, $V_{ik}(t)$ is an increasing linear function of the number of transmitting noninfected links, a decreasing linear function in the number of infected links if the node accepts the information, and $V_{ik}(t)=0$ otherwise. Thus, the function $V_{ik}(t)$ captures the integrity of the information generated at each node by accounting for infected links. 
When the node is not infected, it is transmitting only when in the $\emph{S}$ state.  For a node with degree $k$ and type $i$, let $n_1$ and $n_2$ be the number of links pointing to a node in states $\emph{I}$ and $\emph{S}$, respectively.
The function $V_{ik}(t)$ also captures the integrity of the information generated by the node itself. 
Thus, this function $V_{ik}(t)$  is given by $V_{ik}(t)=n_2-n_1-y_{ik}(t)+(1-y_{ik}(t))=n_2-n_1-y_{ik}(t)+(1-y_{ik}(t))$, where $y_{ik}(t)$ indicates if the attacker successfully injects false information.
Next, we define $\eta(t)$ as the probability that a randomly chosen link is pointing to a susceptible node. $\eta(t)$ will then be given by:
\begin{equation}
\eta(t)=\frac{\sum_k kP(k)\sum_{i \in \mathcal{H}_k}p_k(i) m^{S}_{ik}(t)}{<k>}. \label{etta}
\end{equation}
 Thus, the number of links pointing to nodes in \emph{I}, $\emph{S}$, and either in $\emph{L}$ or $\emph{E}$ states follows a multinomial distribution with parameters $\Theta(t)$, $\eta(t)$, and $1-\Theta(t)-\eta(t)$. 

 Consequently, when the node is susceptible and accepts true information, the expected QoI  is $\bar{V}^T_{ik}(\eta(t))=k\eta(t)+1$. Otherwise, when a susceptible node accepts misinformation, the expected value of $V^M_{ik}(t)$ given that the node receives misinformation and $\boldsymbol{\gamma}(t)=(\Theta(t),\eta(t))$ will be

\vspace{-0.4 cm}

\begin{eqnarray}
&&\hspace{-0.9 cm}\bar{V}^{M}_{ik}(\boldsymbol{\gamma}(t))=\frac{F^M_{ik}(\boldsymbol{\gamma}(t))}{1-L_{ik}(\Theta(t))}, \label{qoIM}
\end{eqnarray}

where
\begin{eqnarray}
&&\hspace{-0.9 cm}\bar{F}^{M}_{ik}(\boldsymbol{\gamma}(t))=\hspace{-0.1 cm}\sum_{n_1=0}^k\sum_{n_2=0}^{k-n_1} \Big(\lambda_{ik} (n_2-n_1-1)+\sum_{n_1=0}^k\sum_{n_2=0}^{k-n_1}(1-\lambda_{ik})  (n_2-n_1)-\sum_{n_2=0}^{k}(1-\lambda_{ik})  (n_2)\Big)\nonumber\\
&&\hspace{1.3 cm}\times\frac{k!}{n_1!n_2!(k-n_1-n_2)!}\times\Theta(t)^{n_1}\eta(t)^{n_2}(1-\Theta(t)-\eta(t))^{k-n_1-n_2} \nonumber\\
&&\hspace{0.7 cm}=k\eta(t)-k\Theta(t)-\lambda_{ik}-(1-\lambda_{ik})k\eta(t).\label{qoIM}
\end{eqnarray}


\normalsize

Whenever a given IoBT node is in the susceptible state and suspects its information to be modified by the attacker, this node becomes exposed, and the QoI in this case will be $\bar{V}^I_{ik}(t)-\delta_{ik}(t)$  if it accepts the information with probability $\beta^E_{ik}(t)$, where $V_{ik}(t)$ is given by (\ref{qoIM}). Otherwise, if it does not accept the information, the QoI will be $0$.

When the node is in the susceptible state and doubts information which is true, the QoI is  $k\eta(t)+1-\delta_{ik}$ if it accepts the information with probability $\beta^L_{ik}$. Thus, the expected QoI given that the node is in the susceptible state with respect to $\alpha_{ik}(t)$ will be given by
\vspace{-0.5 cm}

\begin{eqnarray}
\mathbb{E}_{\alpha_{ik}(t)}[Q_{ik}(t)]&=&\alpha_{ik}(t)(L_{ik}(\Theta(t))(k\eta(t)+1)+\alpha_{ik}(t)(k\eta(t)-k\Theta(t)-\lambda_{ik}-(1-\lambda_{ik})k\eta(t))\nonumber\\
&&\hspace{-0.2 cm}+\sigma_{ik}(t)L_{ik}(\Theta(t))(\beta^L_{ik}(t)(k\eta(t)+1-\kappa \delta_{ik}))\nonumber\\
&&\hspace{-0.2 cm}+\sigma_{ik}(t)(\beta^E_{ik}(k\eta(t)-k\Theta(t)-\lambda_{ik}-(1-\lambda_{ik})k\eta(t)-\kappa \delta_{ik}(1-L_{ik}(\Theta(t))).
\end{eqnarray}
\normalsize

For a node of degree $k$ and type $i$, the cost $c_{ik}$ of infection depends on its importance and functionality in the IoBT. For example, the cost of infection of a fusion center will be higher than that of a cluster head, and the cost of infection of a drone is higher than that of a simple sensor. When the node is in the $\emph{S}$ state, the cost is expressed as the square of the difference between the expected QoI and a target value $Q_T$. Thus, the cost of a node of degree $k$ and type $i$ at time $t$ will be
\vspace{-0.3 cm}
\begin{eqnarray}
&&\hspace{-0.6 cm}v_{ik}(\boldsymbol{x}_{ik},\alpha_{ik}(t), \boldsymbol{\gamma}(t))=x^{S}_{ik}(t)(\mathbb{E}_{\alpha_{ik}(t)}[\bar{Q}_{ik}(\boldsymbol{\gamma}(t))]-Q_T)^2+x^I_{ik}(t)c_{ik}
\label{instcost}
\end{eqnarray}
\normalsize
 The objective of each IoBT node of degree $k$ and type $i$ is to minimize its cost over $[0,T]$, i.e., each IoBT node will seek to solve the following optimization problem:
\begin{eqnarray}
&&\hspace{-1 cm}\min_{\boldsymbol{\alpha}_{ik}(t)}\int_{t=0}^Tv_{ik}({\boldsymbol{x}_{ik},\alpha}_{ik}(t), \boldsymbol{\gamma}(t)) \hspace{0.1 cm} \text{s.t.} \hspace{0.1 cm} {\alpha}_{ik}(t) \in [0,1]^2,   \label{payoff}
\end{eqnarray}
\normalsize
subject to the state constraints in (\ref{xs})-(\ref{xI}) and $\textrm{with}\hspace{0.1 cm} v_{ik}(T)=0$. As shown in (\ref{xs})-(\ref{xI}) and (\ref{instcost}), the state evolution of each IoBT node as well its cost function depend on $\Theta(t)$ and $\eta(t)$ and therefore on the mean-field vector $\boldsymbol{m}_{ik}(t)=(m^l_{ik}(t))_{l \in \mathcal{S}}$ for all $(i,k)$. Further, the dynamics of the IoBT nodes and their cost functons depend on their degree and their type.  Thus, the problem is formulated using game theory \cite{gametheory,gametheory1}. In particular, we use  as a finite state mean-field game \cite{finitemean-field} with multiclass agents \cite{multiclass} as explained next.

\vspace{-0.1 cm}
\section{IoBT Mean-Field Game with Multiclass Agents}
\vspace{-0.1 cm}
\subsection{Game Formulation}
Our problem is formulated as a finite state, mean-field game \cite{finitemean-field} with multiclass agents \cite{multiclass} where the players are the IoBT nodes, and each IoBT node can be in a state belonging to the set $\mathcal{S}=\{S,E,L,I\}$. IoBT nodes having the same degree $k$ and type $i$ belong to class $(i,k)$. We denote by $\mathcal{C}$ the set of all classes.
The proportion $m^l_{ik}$ of IoBT nodes of class $(i,k)$ in each state $l$ evolves according to (\ref{ms})-(\ref{mI}).

We fix a reference player of class $(i,k)$. The state evolution of the reference player is given by  (\ref{xs})-(\ref{xI}). Thus, the state evolution of the reference player depends on $\Theta(t)$ and $\eta(t)$ and therefore on $\boldsymbol{m}_{ik}(t)$ for all $(i,k)$ as well as its control which is the probability $\alpha_{ik}(t)$.
Each IoBT node has full knowledge of the state distribution $\boldsymbol{m}_{ik}(t)$ of the remaining nodes for all classes in $\mathcal{C}$. The cost of each IoBT node of class $(i,k)$ is given by (\ref{instcost}). The objective of each IoBT node of class $(i,k)$ is to find the optimal probability $\alpha_{ik}(t)$ that minimizes its cost according to (\ref{payoff}).

 In order to find the minimum cost, the reference player solves a continuous-time finite-state Markov decision process with finite horizon \cite{finitemarkov} defined by the set of states $\mathcal{S}$. In our game, since each IoBT node chooses the probability $\alpha_{ik}(t)$  when in the \emph{S} state, then the action sets for each state will be given by: $\mathcal{A}_{S}=[0,1]$ and $\mathcal{A}_I=\mathcal{A}_E=\mathcal{A}_L=\phi$. The running costs for each state $l$ are given by
$v_{ik}(S, \boldsymbol{\gamma}(t), \alpha_{ik}(t))=(\mathbb{E}_{\alpha_{ik}(t)}[\bar{Q}_{ik}(\boldsymbol{\gamma}(t))]-Q_T)^2 $, $v_{ik}(E)=v_{ik}(L)=0$, and $v_{ik}(I)=c_{ik}$. Let $u^l_{ik}(t)=\int_{t}^Tv_{ik}(l(s),\boldsymbol{\gamma}(s),\alpha^{j}_{ik}(s))ds$ be the total cost starting from time $t$ when in state $l$ where $l(s)$ is the state at time $s$ and $\alpha^j_{ik}(s)$ is the action taken when in state $j$ at time $s$.  Then, for a given $\boldsymbol{\gamma}(t)$, the reference player uses the so-called Hamilton Jacobi (HJ) equations \cite{finitemean-field} to find the minimum cost. The HJ  equations are defined as
$-\frac{\partial u^l_{ik}}{\partial t}=h(\Delta_l \boldsymbol{u}_{ik}, \boldsymbol{\gamma}(t),l),$
for every state $l \in \mathcal{S}$ where $\boldsymbol{u}_{ik}=(u^l_{ik})_{l \in \mathcal{S}}$, $\Delta_l \boldsymbol{u}_{ik}=(u^j_{ik}-u^l_{ik})_{j \in \mathcal{S}}$, and
\vspace{-0.6 cm}

\small
\begin{eqnarray}
\hspace{-0.5 cm}h(\Delta_l \boldsymbol{u}_{ik}, \boldsymbol{\gamma}(t),l)&=&\min_{\alpha^l_{ik}(t)}v_{ik}(l,\boldsymbol{\gamma}(t),\alpha^l_{ik}(t))+\sum_{j \in \mathcal{S}} G^{ik}_{lj} (\alpha^l_{ik}(t),\Theta(t))(u^j_{ik}-u^l_{ik}), \label{legendre}
\end{eqnarray}
\normalsize
where $G^{ik}_{lj}$ is the transition rate from state $l$ to $j$. 
For our problem, the HJ  equations are specifically given by
\vspace{-1 cm}

\small
\begin{eqnarray}
&&\hspace{-2.9 cm}-\frac{\partial u^{S}_{ik}}{\partial t}=\min_{\alpha_{ik}(t)}v_{ik}(S,\boldsymbol{\gamma}(t),\alpha^{S_T}_{ik}(t))+(1-\alpha_{ik})R_{ik}(\Theta(t))(u^{E}_{ik}-u^{S}_{ik})\hspace{-2 cm}\nonumber\\
&&\hspace{-1.3 cm}+(1-\alpha_{ik})L_{ik}(\Theta(t))(u^{L}_{ik}-u^{S}_{ik})+\alpha_{ik}(t)R_{ik}(\Theta(t))(u^I_{ik}-u^{S_T}_{ik}),
\label{vs}\\
&&\hspace{-3 cm}-\frac{\partial u^{E}_{ik}}{\partial t}=(1-\delta_{ik})\beta^E_{ik}(u^I_{ik}-u^E_{ik})+(1-\delta_{ik})\gamma^E_{ik}(u^S_{ik}-u^E_{ik}),\\
&&\hspace{-3 cm}-\frac{\partial u^{L}_{ik}}{\partial t}=(1-\delta_{ik})(u^S_{ik}-u^L_{ik}),\\
&&\hspace{-3 cm}-\frac{\partial u^I_{ik}}{\partial t}=c_{ik}+\nu_{ik}(u^{S}_{ik}-u^I_{ik}),\label{vi}
\end{eqnarray}
\normalsize
with $u^{S}_{ik}(T)=u^{E}_{ik}(T)=u^{L}_{ik}(T)=u^I_{ik}(T)=0$. Thus, $\alpha_{ik}(t)$  that minimizes the Hamiltonian $h(\Delta_{S} \boldsymbol{u}_{ik}, \boldsymbol{\gamma}(t),S)$ is the optimal value. We denote by $\alpha_{ik}(\Delta_{S}\boldsymbol{u}_{ik}, \boldsymbol{\gamma}(t))$ the optimal value.  Due to the dependence of the optimal value on the mean-field dynamics through $\boldsymbol{\gamma}(t)$, the optimal value $\alpha_{ik}(\Delta_{S}\boldsymbol{u}_{ik},  \boldsymbol{\gamma}(t))$ is called the \emph{best response} with respect to $\boldsymbol{\gamma}(t)$. The following remark presents the best response $\alpha_{ik}(\Delta_{S}\boldsymbol{u}_{ik}, \boldsymbol{\gamma}(t))$  of a player of class $(i,k)$ for $\boldsymbol{\gamma}(t)$.

\begin{remark}
\emph{For a given $\boldsymbol{\gamma}(t)$, the best response $\alpha_{ik}(\Delta_{l}\boldsymbol{u}_{ik},\boldsymbol{\gamma}(t))$ of a player of class $(i,k)$ is }
\small
\begin{equation}
  \hspace{-0.1 cm}  \alpha_{ik}(\Delta_l\boldsymbol{u}_{ik},\boldsymbol{\gamma}(t))=
\begin{cases}
 0, &\hspace{-0.2 cm} \text{\emph{if} }  g_{ik}(\Delta_l\boldsymbol{u}_{ik},\boldsymbol{\gamma}(t)) < 0,\\
\smaller g_{ik}(\Delta_S\boldsymbol{u}_{ik}\normalsize,\boldsymbol{\gamma}(t)), & \hspace{-0.3 cm} \text{\emph{if} }  0<g_{ik}(\Delta_{l}\boldsymbol{u}_{ik},\boldsymbol{\gamma}(t)) < 1,\\
  1,              & \text{\emph{otherwise}},
\end{cases} \label{alphab}
\end{equation}
\normalsize
\emph{where}
\vspace{-0.2 cm}
\small
\begin{equation}
g_{ik}(\Delta_{l}\boldsymbol{u}_{ik},\boldsymbol{\gamma}(t))=\frac{R_{ik}(\Theta(t))(u^E_{ik}-u^S_{ik})+L_{ik}(\Theta(t))(u^L_{ik}-u^S_{ik})-R_{ik}(\Theta(t))(u^I_{ik}-u^S_{ik})+2A_1(Q_T-A_2)}{2A^2_1},\nonumber\\
\end{equation}
\begin{eqnarray}
&&\hspace{-0.8 cm}A1=(L_{ik}(\Theta(t))(k\eta(t)+1)+(k\eta(t)-k\Theta(t)-\lambda_{ik}-(1-\lambda_{ik})k\eta(t)))+L_{ik}(\Theta(t))(\beta^L_{ik}(t)(k\eta(t)+1-\kappa\delta_{ik}))\nonumber\\
&&-(\beta^E_{ik}(k\eta(t)-k\Theta(t)-\lambda_{ik}-(1-\lambda_{ik})k\eta(t)-\kappa\delta_{ik}(1-L_{ik}(\Theta(t))), \nonumber
\end{eqnarray}
\begin{eqnarray}
&& \hspace{-0.9 cm}A_2=+L_{ik}(\Theta(t))(\beta^L_{ik}(t)(k\eta(t)+1-\kappa\delta_{ik}))+(1-L_{ik}(\Theta(t)))(\beta^E_{ik}(k\eta(t)-k\Theta(t)-\lambda_{ik}-(1-\lambda_{ik})k\eta(t)-\kappa \delta_{ik})).\nonumber
\end{eqnarray}
\end{remark}
\vspace{-0.3 cm}
The result directly follows by equating the partial derivative of the right-hand side of (\ref{vs})  with respect to $\alpha_{ik}(t)$. Given the best response in (\ref{alphab}), next, we characterize the MFE for our IoBT game .
\vspace{-0.5 cm}
\subsection{Mean-Field Equilibrium}
The MFE occurs when the best response $\alpha_{ik}(\Delta_S\boldsymbol{u}_{ik},\boldsymbol{\gamma}(t))$ of a player belonging to class $(i,k)$ is the same as the strategy $\bar{\alpha}_{ik}$ of the population belonging to class $(i,k)$, for all $(i,k)$. Thus, the MFE will be the solution of the Hamilton-Jacobi equations in (\ref{vs})-(\ref{vi}), and the Kolomogrov equations in (\ref{ms})-(\ref{mI}) with $\bar{\alpha}_{ik}(t)=\alpha_{ik}(\Delta_S\boldsymbol{u}_{ik},\boldsymbol{\gamma}(t))$, $u^{S}_{ik}(T)=u^{E}_{ik}(T)=u^{L}_{ik}(T)=u^I_{ik}(T)=0$ and $m^{S}_{ik}(0)=1$, $m^{L}_{ik}(0)=m^{E}_{ik}(0)=m^I_{ik}(0)=0 \hspace{0.2 cm} \forall i,k.$ For our IoBT mean-field game, the MFE exists. The proof follows from \cite[Proposition 4]{finitemean-field}. Even though existence follows from \cite{finitemean-field}, it is difficult to characterize analytically the MFE for  our IoBT game. Further, our game does not satisfy the standard conditions for uniqueness for mean-field games (see \cite{finitemean-field} and \cite{uniqueness}). In particular, the Hamiltonian $h(\Delta_S \boldsymbol{u}_{ik},\boldsymbol{\gamma}(t),l)$  is not strongly convex in $\Delta_s \boldsymbol{u}_{ik}$.
Thus, due to the aformentioned reasons, it is difficult to analyze analytically the uniqueness of our game.

In (\ref{ms})-(\ref{mI}), the mean-field equations are subject to initial conditions. Thus, in order to find the MFE, we propose an algorithm based on the forward backward sweep method \cite{FBSM}, which has been widely used to solve optimal control problems with initial conditions. The details of the proposed forward backward sweep as tailored to the IoBT are given by Algorithm \ref{FBSMalg}.


\begin{algorithm}[t]
\smaller
 \textbf{Input:} $\epsilon$, $T$, $\nu_{ik}$, $\lambda_{ik}$ for all $(i,k)$ \\
 \textbf{Output:} The equilibrium acceptance probabilities  $\boldsymbol{\alpha^{*}_{ik}}=(\alpha^{*}_{ik}(t))_{t \in [0,T]}$ for all $(i,k)$\\
Initialize: for all $t \in [0,T]$, $\alpha^{*}_{ik}(t)=\alpha_{ik,0},$ $iter=0$ \\
\Repeat{ $||\boldsymbol{\alpha^{*}_{ik}}-\boldsymbol{\alpha^{*}_{ik,\textrm{old}}}|| \leq \epsilon$ or $iter >I_{\max}$ }{
$\boldsymbol{\alpha^{*}_{ik,\textrm{old}}} \leftarrow \boldsymbol{\alpha^{*}_{ik}}$ \\
Compute $m^*_{ik}(t)$ using the mean-field equations (\ref{ms})-(\ref{mI}) with $\bar{\alpha}_{ik}(t)=\alpha_{ik,\textrm{old}}(t)$ $\forall (i,k)$, $t\in[0,t]$ (forward sweep)\\
Using  $m^*_{ik}(t)$, (\ref{theta}), and (\ref{etta}), compute  $\boldsymbol{\gamma}^*(t)$ $\forall t \in [0,T]$\\
Using $\boldsymbol{\gamma}^*(t)_{t \in [0,T]}$ , compute $\boldsymbol{\alpha^{*}_{ik}}$  using the Hamilton Jacobi equations (\ref{vs})-(\ref{vi})
(backward sweep)
$iter=iter+1$

}
\caption{Forward Backward Sweep Algorithm for the IoBT mean-field Game}
\label{FBSMalg}
\end{algorithm}
Algorithm \ref{FBSMalg} first solves the mean-field equations using a finite difference method and using the initial guess of the optimal acceptance probability $\boldsymbol{\alpha_{ik}}$. Next, using the mean-field solution $(m^*_{ik}(t))_{t \in [0,T]}$, the probabilities $\Theta^*(t)$ and $\eta^*(t)$  are computed according to (\ref{theta}) and (\ref{etta}), respectively. Algorithm \ref{FBSMalg} then computes the optimal acceptance probability $\boldsymbol{\alpha_{ik}}$  based on the HJ equations and on the the computed $(\Theta^*(t))_{t \in [0,T]}$ and  $(\eta^*(t))_{t \in [0,T]}$, using a finite difference method. The newly computed $\boldsymbol{\alpha_{ik}}$  are consequently used to recompute the mean-field solution. The process is repeated until convergence, or if the maximum number of iterations $I_{\max}$ is reached, since the forward backward sweep algorithm is not guaranteed to converge, in general. 

\vspace{-0.1 cm}
In practice, each IoBT node will run Algorithm 1 in order to determine its optimal probability of accepting information $\alpha^*_{ik}(t)$ over the considered interval $[0,T]$. It is assumed that the IoBT nodes acquire information about the characteristics (such as the processing delay $\delta_{ik}$ and the probabilities $\beta^E_{ik}$ and $\beta^L_{ik}$) of IoBT nodes belonging to other classes through an initial phase in which the fusion center/ base station broadcasts the information. For each subsequent time duration, each IoBT node will run Algorithm 1 and use, as the initial mean-field values, the mean-field values $m^{l*}_{ik}(T)$ at time $T$ from the previous run.
The process will be repeated until the end of the military operation.

\normalsize

While the mean-field formulation provides a tractable approach to analyze a massive IoBT system, it assumes that the number of IoBT devices to be infinite. In practice, the IoBT will have a large, but finite number of devices. As such, in the next section, we analyze an IoBT having a large, but finite number of nodes. The finite IoBT case is a better fit to a real-world IoBT and can account for all potential network sizes. However, computing the equilibria for the finite case is computationaly expensive as discussed next, and therefore is not suitable to be used by the IoBT nodes to determine the optimal strategies in real-time.  Thus, the complexity of determining the equilibria of the finite IoBT makes the proposed mean-field approach in Section II more suitable to determine the optimal stategies of the IoBT nodes.
\vspace{-0.3 cm}
\section{On the Tractability of the  Finite IoBT Case}
\vspace{-0.2 cm}
\subsection{Game Formulation}
We now consider the game when the IoBT network is composed of a finite number of nodes $N+1$, each of which can be in any one of the states in $\mathcal{S}$. The number of \emph{background nodes}, i.e. any node other than the reference player, belonging to class $(i,k)$ is $N_{ik}$ with the assumption that $\lim_{N\to\infty}\frac{N_{ik}}{N}=\pi_{ik}$ where $\pi_{ik}=P(k)p_k(i)$ is the distribution of the players in the previously studied infinite IoBT case. 

Thus, at any time $t$, the number of background nodes $\boldsymbol{n}_{ik}(t)=(n^j_{ik}(t))_{j \in \mathcal{S}}$ of class $(i,k)$ in each state  evolves according to a Markov chain.  
Hence, the probability that  a link is infected in the finite IoBT is
$\Theta_N(t)=\sum_k k \sum_{i \in \mathcal{H}_k} \frac{n^I_{ik}(t)}{N}$
Similarly, the probability of a randomly chosen link pointing to a susceptible node is
$\eta_N(t)=\sum_k k \sum_{i \in \mathcal{H}_k} \frac{n^{S}_{ik}(t)}{N}$.
We define the vector $\boldsymbol{n}(t)=(\boldsymbol{n}_{ik}(t))_{(i,k) \in \mathcal{C}}$ .
The state evolution of a background node belonging to class $(i,k)$ is 
\vspace{-0.4 cm}

\begin{equation}
   \mathbb{P}(s_{ik}(t+h)=j|s_{ik}(t)=l)= G^{lj}_{ik}(\Theta_N(t),\bar{\alpha}^l_{ik}(t))h+o(h), 
\label{finiteevol}
\end{equation}
\normalsize
where $o(h) \rightarrow 0$ as $h \rightarrow \infty$, $G^{lj}_{ik}$ is the transition rate from state $j$ to state $l$ and has the same expression as the infinite case with the variable $\Theta(t)$ replaced by $\Theta_N(t)$, and $\bar{\alpha}^{N}_{ik}(t)$ is the acceptance probability of the background nodes of class $(i,k)$. Thus, in our problem, $G^{lj}_{ik}$ depends only on $\Theta_N(t)$ and $\bar{\alpha}^{N}_{ik}(t)$ when it is in the $S$ state. Subsequently, $G^{lj}_{ik}(\Theta_N(t), \bar{\alpha}_{ik}(t))$ is replaced by $G^{lj}_{ik}(\boldsymbol{n}(t), \bar{\alpha}^N_{ik}(t))$   as $\Theta_N(t)$ is a  function of $n^I_{ik}(t)$ $\forall (i,k) \in \mathcal{C}$. Thus, the evolution of the number of nodes $\boldsymbol{n}_{ik}(t)$ will affect the transition rate $G^{lj}_{ik}$. The state evolution of the reference node is also given by (\ref{finiteevol}) but with $\bar{\alpha}^{N}_{ik}(t)$ replaced by $\alpha^{N}_{ik}(t)$.
Further, the state transitional rates of the different nodes are independent conditioned on $\boldsymbol{n}(t)$ and $l(t)$, where $l(t)$ is the state of the reference node. Thus, the evolution of $\boldsymbol{n}_{ik}(t)$ is given by 
$\mathbb{P}(\boldsymbol{n}_{ik}(t+h)=\boldsymbol{n}_{ik}+\boldsymbol{e}_{jl}|\boldsymbol{n}_{ik}(t)=\boldsymbol{n_{ik}}, \boldsymbol{n}(t)=\boldsymbol{n}, \boldsymbol{s}(t)=s)=\eta^s_{ik}(j,l,\boldsymbol{n})h+o(h)$ 
where  $\boldsymbol{e}_{jl}=\boldsymbol{e}_{j}-\boldsymbol{e}_{l}$, $\boldsymbol{e}_{j}$ is the $j^{th}$ vector of the canonical basis of $\mathbb{R}^{|\mathcal{S}|}$, 
and
\vspace{-0.1 cm}
\begin{equation}
   \rho^s_{ik}(l,j,\boldsymbol{n})= 
\begin{cases}
  n^l_{ik} G^{lj}_{ik}(\boldsymbol{n}',\bar{\alpha}^{N}_{ik}(t)), & \text{if } l=S,\\
   n^l_{ik} G^{lj}_{ik},              & \text{otherwise},
\end{cases} \label{rho}
\end{equation}
where $\boldsymbol{n}'=\boldsymbol{n}+\boldsymbol{e}^{ik}_{sl}$
if $(i,k)=(i',k')$ where $(i',k')$ is the class of the reference player, $\boldsymbol{n}+\boldsymbol{e}^{ik}_{sl}$ is the same as $\boldsymbol{n}$ but with $\boldsymbol{n}_{ik}$ replaced by $\boldsymbol{n}_{ik}+\boldsymbol{e}_{sl}$. $\boldsymbol{n}'=\boldsymbol{n}-\boldsymbol{e}^{ik}_{l}$ if $(i,k) \neq (i',k')$ where $\boldsymbol{n}-\boldsymbol{e}^{ik}_{l}$ is the same as $\boldsymbol{n}$ but with $\boldsymbol{n}_{ik}$ replaced by $\boldsymbol{n}_{ik}-\boldsymbol{e}_{l}$.



The reference node has full knowledge of the evolution of $\boldsymbol{n}_{ik}(t)=(n^j_{ik}(t))_{j \in \mathcal{S}}$ of the background nodes for all $(i,k) \in \mathcal{C}$. Thus, the reference IoBT node of class $(i,k)$ starting from state $l$ seeks to find the acceptance probabilities $\alpha^{N,l}_{ik}(t)$ ($l \in \{S_N,S_T\}$) that minimizes its total expected cost 
$u^{N,\boldsymbol{n},l}_{ik}=\min_{\alpha^N_{ik}(t)}\mathbb{E}\int_{0}^T v_{ik}(j(s),\boldsymbol{\gamma}^N(s),\alpha^j_{ik}(s))ds $
where $\boldsymbol{\gamma}^N(s)=(\Theta^N(s),\eta^N(s))$ and $v_{ik}(j(s),\boldsymbol{\gamma}^N(s),\alpha^j_{ik}(s))$ has the same expression as the infinite  case .
We denote by $u^{N,\boldsymbol{n},l}_{ik}(t)$ the total expected cost starting from time $t$ and when in state $l$ conditioned on $\boldsymbol{n}(t)=\boldsymbol{n}$. The HJ equation for this case will be given by  \cite{finitemean-field}
\vspace{-0.4 cm}

\small
\begin{eqnarray}
&&\hspace{-2 cm}-\frac{du^{N,\boldsymbol{n},l}_{ik}}{dt}=\sum_{r,v}\rho^l_{ik}(v,r,\boldsymbol{n})(u^{N,\boldsymbol{n}+e^{ik}_{vr},l}_{ik} - u^{N,\boldsymbol{n},l}_{ik}) + h(\Delta_l \boldsymbol{u}^{N,\boldsymbol{n}}_{ik} ,\boldsymbol{\gamma}_N(t), l),\label{HBf}
\end{eqnarray}
\normalsize
\vspace{-0.4 cm}
\\
with $u^{N,\boldsymbol{n},l}_{ik}(T)=0$ for all $l \in \mathcal{S}$ and  $(i,k) \in \mathcal{C}$.
where $\Delta_l \boldsymbol{u}^{N,\boldsymbol{n}}_{ik}=(u^{N,\boldsymbol{n},j}_{ik}-u^{N,\boldsymbol{n},l}_{ik})_{j \in \mathcal{S}}$ and $h(\Delta_l u^{N,\boldsymbol{n}}_{ik} ,\boldsymbol{\gamma}_N(t), l)$ has the same expression as in  (\ref{legendre}).
%

In the finite IoBT game, the equilibrium also occurs when $\alpha^{N}_{ik}(t)=\bar{\alpha}^{N}_{ik}(t)$ $\forall$ $(i,k)$. It can be easily shown that the equilibrium exists using a similar argument as in \cite{finitemean-field}.
However, due to the dependence of $u^{N,\boldsymbol{n},l}_{ik}$ on $\boldsymbol{n}$ in (\ref{HBf}), the number of possible evaluations of $u^{N,\boldsymbol{n},l}_{ik}$  grows exponentially with $T$ according to (\ref{HBf}). Hence, computing $u^{N,\boldsymbol{n},l}_{ik}$ for the finite IoBT is computationally expensive, unlike in the proposed mean-field approach where the computational complexity of computing $u^{l}_{ik}$ for all $l \in \mathcal{S}$ is only linear in $T$ according to (\ref{vs})-(\ref{alphab}). Thus, the mean-field approach is computationally more favourable to be used in terms of to find the optimal strategies of the IoBT nodes.
However, in order to ensure that the mean-field game yields a performance comparable to the finite IoBT game for large number of nodes, we discuss the convergence of the cost and distribution functions of the finite IoBT case to the mean-field case, as $N$ goes to infinity. 
\vspace{-0.4 cm}
\subsection{Convergence Conditions of the Finite IoBT Game}
\vspace{-0.2 cm}
We now extend the convergence results of the finite state mean-field games in \cite{finitemean-field} to the case of multiclass agents and when the transitional probability is a function of the control as well as the mean-field. We show the conditions under which the cost and distribution functions of the $N+1$ player game converges uniformly to corresponding functions in the mean-field game, in order to ensure that the mean-field IoBT game yields a performance comparable to the finite IoBT game.  Our proof relies on the following useful property from \cite[Proposition 7]{finitemean-field} which holds for the solution $u^{N,\boldsymbol{n},l}_{ik}$ to our HJ equations in (\ref{HBf}). 
\begin{remark}
\emph{Let $u^{N,\boldsymbol{n},l}_{ik}(t)$ be the solution of (\ref{HBf}). Then, there exists $C>0$ and $T^*>0$ such that for $0<T<T^*$, }
\vspace{-0.7 cm}

\small
\begin{equation}
\max_{rv}||u^{N,\boldsymbol{n}+\boldsymbol{e}^{rv}_{ik},l}_{ik}(t)-u^{N,\boldsymbol{n},l}_{ik}(t)|| \leq \frac{2C}{N},
\end{equation}
\label{grad}
\end{remark}
\vspace{-0.7 cm}
where the norm ||.|| used is the $\infty$ norm.

Also, we rely on the following properties of our game.

\begin{prop}
The studied IoBT game exhibits the following properties:
\begin{enumerate}
\item The transitional rate $G^{jl}_{ik}(\alpha_{ik}(t), \Theta(t))$ is a Lipchitz function of  $\alpha_{ik}(t)$ for all $i,k$.
\item The best response $\alpha^*_{ik}(\Delta_S \boldsymbol{u}_{ik}, \boldsymbol{\gamma}(t))$ is Lipschitz in $\Delta_l \boldsymbol{u}_{ik}$, $\Theta(t)$, $\eta(t)$, and $m_{ik}(t)$ $\forall$ $(i,k) \in \mathcal{C}$ provided that the immediate cost $v_{ik}$ is strongly convex in $\alpha_{ik}$.
\item   The transitional rate $G^{jl}_{ik}(\alpha^*_{ik},\Theta(t))$ is Lipchitz in  $\Delta_j u$ and $\Theta(t)$.
\item The immediate cost $v_{ik}$  and its derivative $\nabla_{\alpha_{ik}} v_{ik}$ is Lipchitz in $\Theta(t)$ and $\eta(t)$.
\item  The function $h(\Delta_S \boldsymbol{u}_{ik},\boldsymbol{\gamma}(t),l)$ is Lipchitz in $\Delta_l u$, $\Theta(t)$, and $\eta(t)$.
\end{enumerate}
\label{prop1}
\end{prop}

\begin{proof}
See Appendix A.
\end{proof}

Using the results of Proposition \ref{prop1}, we can now present the convergence results in the following theorem.

\begin{thm}
Let $T^*$ be defined as in Remark \ref{grad}. There exists a constant $C$ independent of $N$, for which, if $T<T^*$, satisfies $\mu=TC<1$ then
\vspace{-0.4 cm}

\small
\begin{equation}
\sum_{i,k}V^N_{ik}(t)+W^N_{ik}(t) \leq \frac{C}{1-\mu}\frac{1}{N_{\max}},
\end{equation}\normalsize
\vspace{-0.4 cm}
\\
for all $t \in [0,T]$, where $N_{\max}=\max_{(r,v)\in\mathcal{C}}N_{rv}$, $W^N_{ik}(t)=\mathbb{E}\Big[||\boldsymbol{u}_{ik}(t)-\boldsymbol{u}^{N,\boldsymbol{n}}_{ik}(t)||^2\Big]$, $V^N_{ik}(t)=\mathbb{E}(||\frac{\boldsymbol{n}_{ik}(t)}{N_{ik}}-\boldsymbol{m}_{ik}(t)||)^2$, $\boldsymbol{m}_{ik}(t)$ and $\boldsymbol{u}_{ik}(t)$ are the mean-field and cost functions at the MFE, $\boldsymbol{n}_{ik}(t)$ and $\boldsymbol{u}^{N,\boldsymbol{n}}_{ik}(t)$ are the equilibrium distribution and cost value of $N+1$ player game .
\end{thm}

\begin{proof}
The proof of Theorem 1 relies on the following two lemmas.
\begin{lem}
Define $T^*$ defined as  done in Remark \ref{grad}, then, there exists $C_1$ such that
\begin{equation}
W^N_{ik}(t)\leq \frac{C_1}{N}+C_1 \mathbb{E} \int_{t}^T \Big(W^N_{ik}(s) + \sum_{(r,v) \in \mathcal{C}}V^N_{rv}(s)\Big)ds.
\label{Wik}
\end{equation}
\end{lem}
\begin{proof}
See Appendix B.
\end{proof}
\begin{lem}
Define $T^*$ as done in Remark \ref{grad}, then, there exists $C_2$ such that
\begin{equation}
V^N_{ik}(t) \leq C_2 \mathbb{E} \int_{0}^t (V^N_{ik}(s)+W^N_{ik}(s)+V^N_{yz}(s)) ds +\frac{C_2}{N_{\max}}, \label{vik}
\end{equation}
where $(y,z)=\arg \max_{(r,v)} V_{rv}(t)$ and $N_{\max}=\max_{(r,v)\in \mathcal{C}}N_{rv}$.
\end{lem}
\begin{proof}
See Appendix C.
\end{proof}
By adding (\ref{Wik})  and (\ref{vik}) for all $(i,k)$, we have
\small
\begin{eqnarray}
\sum_{ik}W^N_{ik}(t)+\sum_{ik} V^N_{ik}(t) &\leq& C_1 \mathbb{E} \int_{0}^t \sum_{ik}\Big(W^N_{ik}(s)+\sum_{rv}V^N_{rv}(s)\Big)
+\frac{C_1 |\mathcal{C}|}{N}ds \nonumber\\
&&+ C_2 \mathbb{E} \int_{t}^T \sum_{ik}\Big(W^N_{ik}(s)+V^N_{ik}(s)+V^N_{yz}(s)\Big)+ \frac{C_2|\mathcal{C}|}{N_{\max}},\nonumber\\
&\leq& \bar{C} \mathbb{E} \int_{0}^T \sum_{ik}V^N_{ik}(s)+W^N_{ik}(s) + \frac{\bar{C}}{N_{\max}},
\end{eqnarray}
\normalsize
where $\bar{C}=\max\{C_1|\mathcal{C}|,C_2+1, C_2|\mathcal{C}|\}$. \\
Let $W^N_{ik}+V^N_{ik} =\max_{0 \leq t \leq T}W^N_{ik}(t)+V^N_{ik}(t)$. Then,
\begin{eqnarray}
\sum_{ik}W^N_{ik}(t)+V^N_{ik}(t) \leq \sum_{ik} W^N_{ik}+V^N_{ik} \leq \bar{C}T\sum_{ik}W^N_{ik}+V^N_{ik}+\frac{\bar{C}}{N_{\max}} \leq \frac{\bar{C}}{(1-\mu)N_{\max}},
\end{eqnarray}

where $\mu=\bar{C}T$. Thus, the value function and the proportion of nodes converges uniformly in distribution to the meanfield case. Hence, the mean-field equilibrium constitutes an $\epsilon$ equilibrium for the finite game as demonstrated in \cite{mean-fieldSIR}.
\end{proof}
\vspace{-0.3 cm}
Thus in this section, we have demonstrated that finding the equilibrium for the finite game has exponential complexity in the number of IoBT nodes. However, we have also show that, under mild conditions, the finite game will converge to the mean-field game as $N$ goes to infinity.

\vspace{-0.6 cm}

\section{Simulation Results}
\vspace{-0.1 cm}
For our simulations, we consider an IoBT where the values of the degree of the nodes $k \in \{1,10, 15, 20\}$. The degree distribution is $
P(k=1)=0.4$, $P(k=10)=0.3$, $P(k=15)=0.2$, and $P(k=20)=0.1$. 
The distribution is chosen such that the proportion of nodes decreases with the degree, which represents a typical hierarchical IoBT structure. We consider one type of device of each degree. Thus, nodes of degree $1$ correspond to simple sensors. Nodes of degrees $10$ and $15$ correspond to cluster heads, and nodes of degree $15$ correspond to local sinks. The infection costs are set to:
$c_{1}=1$, $c_{10}=10$, $c_{15}=20$,  and $c_{20}=30$. The cost values are chosen depending on the importance of the nodes.  For a node of degree $k$,  the target QoI $Q_T=k$. The attacker's infection rate for all nodes  is set to be $0.2$. The time period $T$ is set to be $0.9$ seconds. 
The delays of the $\emph{E}$ and $\emph{L}$ states are set to be
$\delta_1=0$, $\delta_{10}=0.4$, $\delta_{15}=0.3$, and $\delta_{20}=0.3$. The acceptance probabilities of the $\emph{E}$ state are set to $\beta^E_1=0.5$, $\beta^E_{10}=0.3$, $\beta^E_{15}=0.2$, and $\beta^E_{20}=0.1$. The acceptance probabilities when in $\emph{L}$ state are set to $\beta^L_1=0.5$, $\beta^L_{10}=0.6$, $\beta^L_{15}=0.7$, and $\beta^L_{20}=0.8$. The parameters are chosen to reflect the computational capabilities of the different IoBT nodes. For example, $\beta^E_1=\beta^L_1=0.5$ is chosen for simple sensors that can not identify the misinformation. Thus, such sensors accept/reject the information with probability $0.5$, and, hence, the processing delay $\delta_1$ is set to $0$.
For comparison, we consider a baseline in which the nodes always transmit with probability one.
For the considered simulation values, we compute the equilibrium acceptance probability using Algorithm 1. 
We also compute the proportion of infected nodes, the probability of an infected link, and the QoI for both the baseline and the MFE. For all considered values, Algorithm 1 converges in at most $16$ iterations. Also, Algorithm 1 always yields the same solution for any initial guess of the acceptance probabilities.

\begin{figure}[t]{
	\centering
	\includegraphics[width=9 cm,height=5.3cm,angle=0]{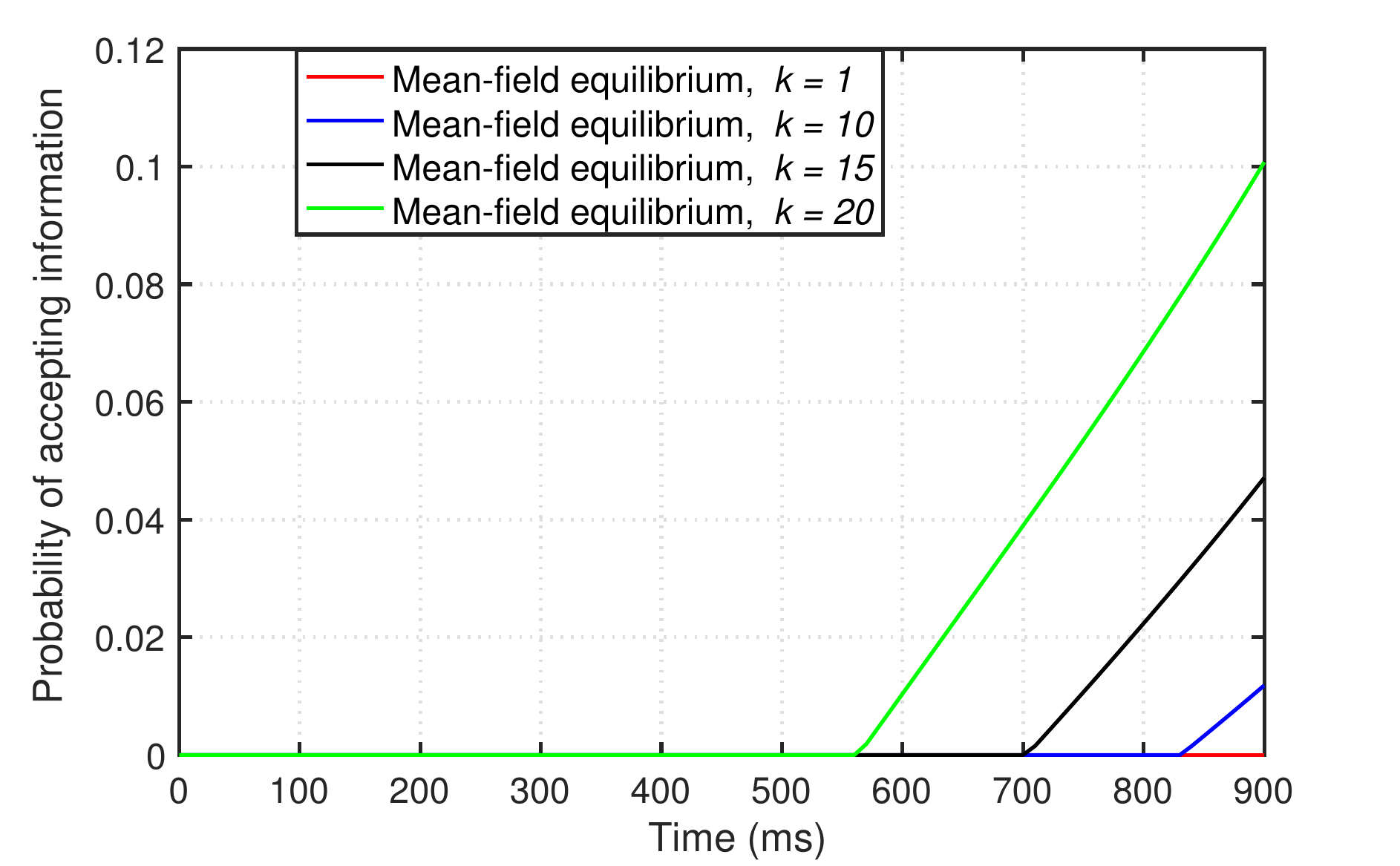}
	\caption{The acceptance probability  as a function of time.}\label{acceptance}
	}\vspace{-0.4 cm}
\end{figure}

\begin{figure}[t]{
	\centering
	\includegraphics[width=9 cm,height=5.3cm,angle=0]{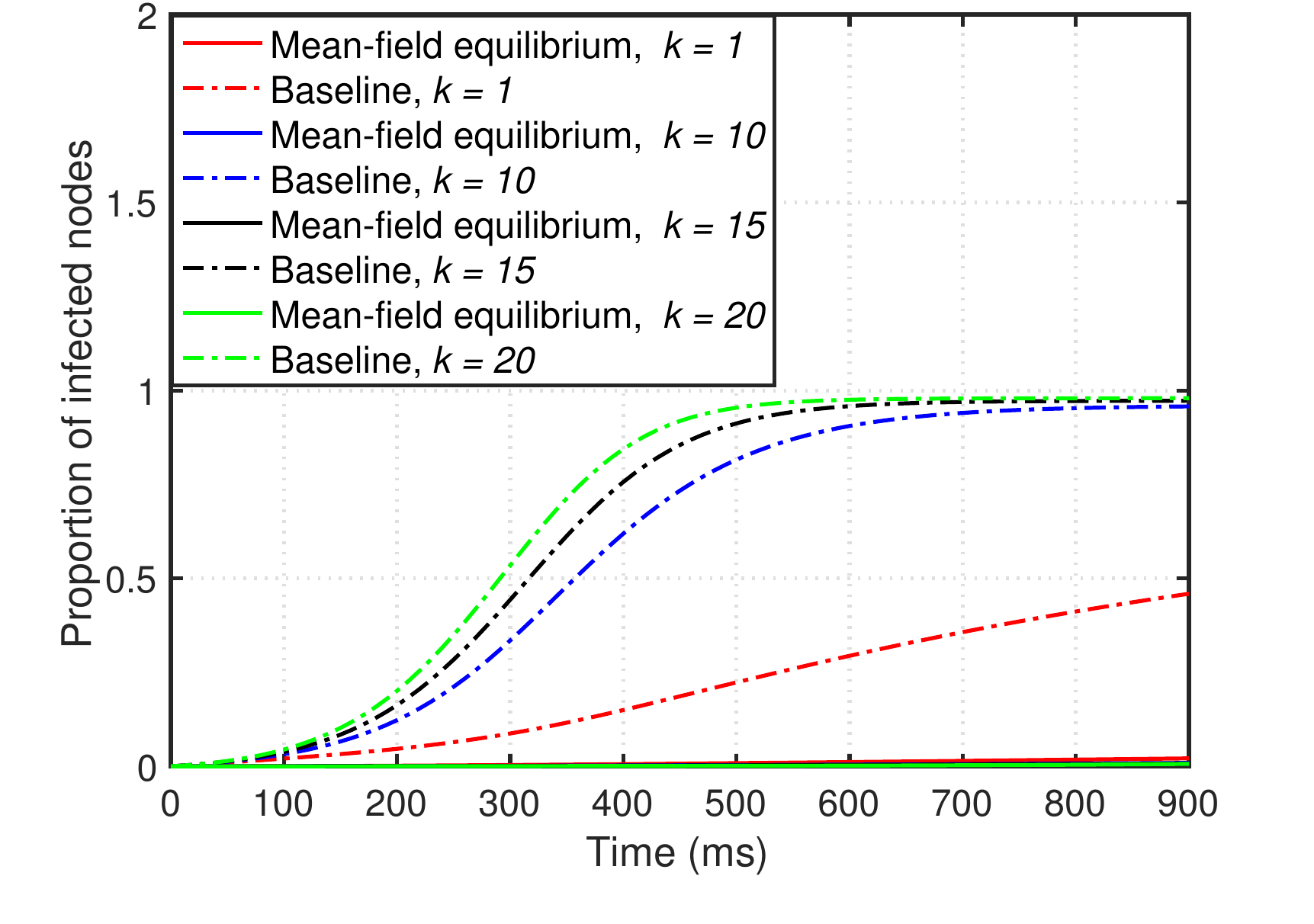}
	\caption{Evolution of the proportion of infected nodes  over time.}\label{infected}
	}\vspace{-0.4 cm}
\end{figure}

\begin{figure}[t]{
	\centering
	\includegraphics[width=9 cm,height=5.3 cm,angle=0]{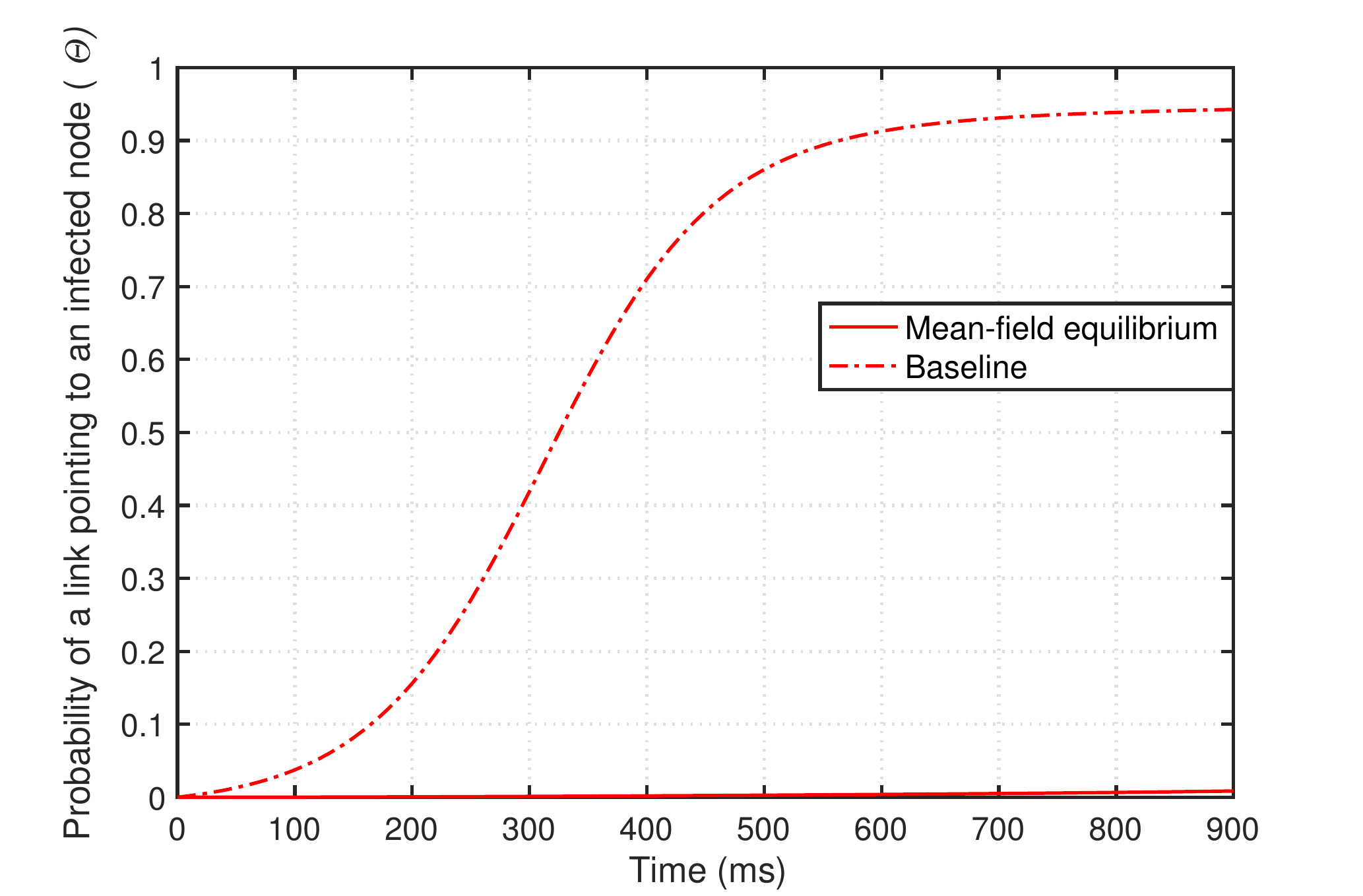}
	\caption{Evolution of the probability of an infected link  over time.}\label{theta}
	}\vspace{-0.7 cm}
\end{figure}

\begin{figure}[t]{
	\centering
	\includegraphics[width=9 cm,height=5.3cm,angle=0]{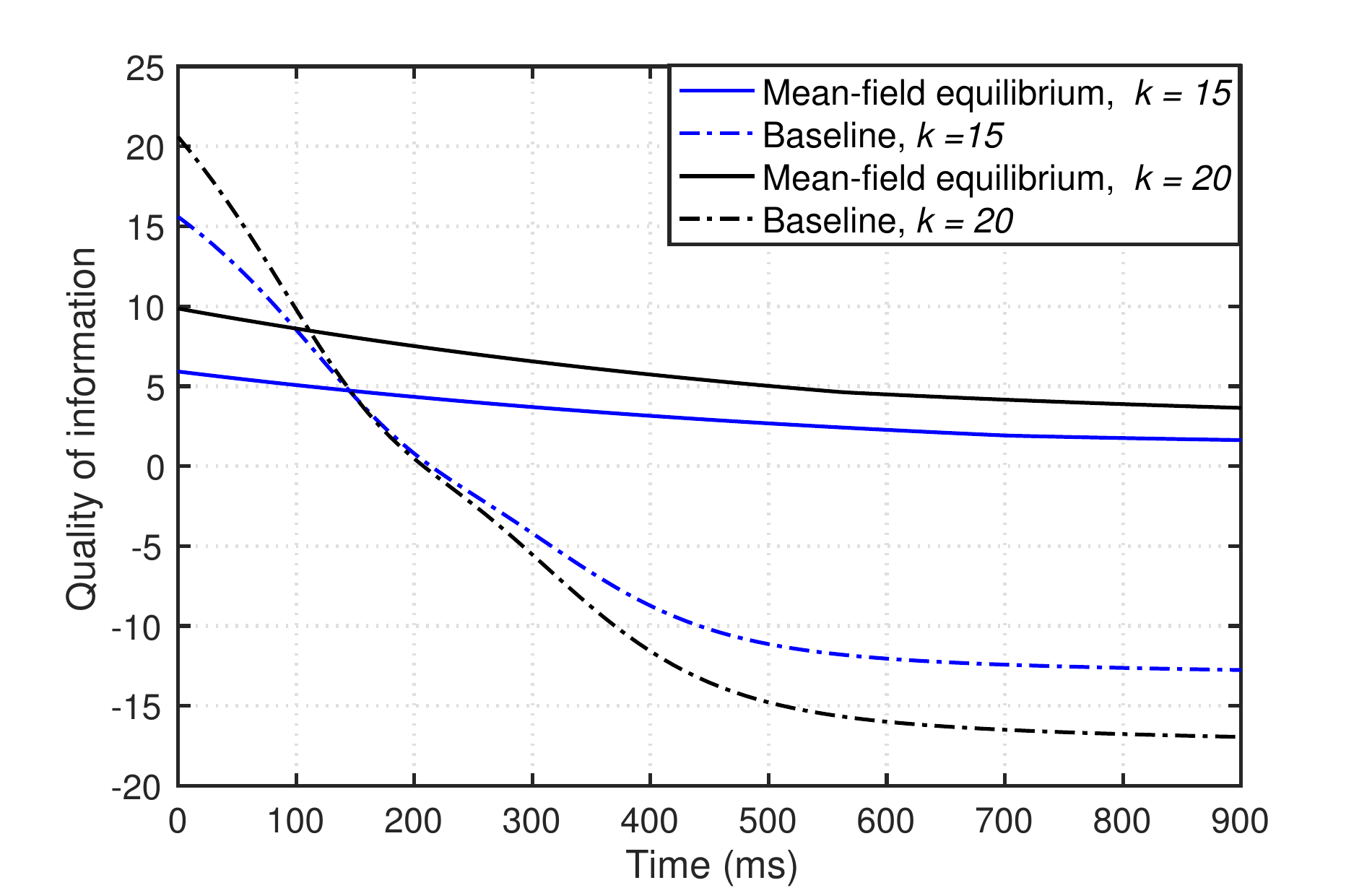}
	\caption{Time evolution of the QoI resulting from the proposed MFE and the baseline.}\label{qoI}
	}\vspace{-0.3 cm}
\end{figure}

Fig. \ref{acceptance} shows the MFE acceptance probability $\alpha$ versus time for the considered values of degree $k$. First, when $k=1$, the acceptance probability is zero for the entire time duration, since the processing delay $\delta_1$ is $0$. Thus, a node of degree $1$ can reduce the spread of misinformation by accepting the information with probability $\beta^E_1=0.5$ instead of $1$.
When $k=10$, the acceptance probability is $0$ for $t \leq 0.83$ seconds. Then, it increases with time until it reaches 0.01 at $t=0.9$, as the spread of misinformation ceases in the IoBT.
When $k=15, 20$, the acceptance probability varies similar to the case when $k=10$.


Fig. \ref{infected} shows how the proportion of infected nodes changes over time for for the considered values of degree $k$ and for both the baseline as well as the MFE. Using the baseline and for the considered degree values, the  proportion of infected nodes increases with time until it reaches $0.45$, $0.95$, $0.97$, and $0.98$ at $t=0.9$ for $k=$$\hspace{0.1 cm}1$, $10$, $15$, and $20$, respectively.
From Fig. \ref{infected} we can also see that, using the MFE and for all considered degree values,  the  proportion of infected nodes increases with time until it reaches $0.0212$, $0.009$, $0.0078$, and $0.0065$ at $t=0.9$ for $k=$$\hspace{0.1 cm}1$, $10$, $15$, and $20$, respectively. Thus, Fig. \ref{infected} shows that, for all considered degree values, the proportion of infected nodes using the MFE is maintained significantly lower than the baseline case.  The considerable decrease in the the proportion of infected nodes is due to two reasons 1) At the MFE, the acceptance probabilities of information for all nodes when in the \emph{S} is zero for a considerable time duration as shown in Fig. \ref{acceptance} 2) The acceptance probabilities of misinformation when in the $\emph{E}$ are low, which limits the spread of misinformation.  The decrease in the proportion of infected nodes reaches up to $99\%$ when $k=15$.

In Fig. \ref{theta}, we show the probability of an infected link $\Theta$ over time for both the baseline and the MFE. From this figure, we can see that, for the baseline, $\Theta$ increases with time until it reaches $0.94$ at $t=0.9$ seconds. This is due to the fact that the proportion of infected nodes increases with time using the baseline, for all considered degree values as shown in Fig. \ref{infected}. For the MFE, from Fig. \ref{theta} we can see that $\Theta$  increases with time until it reaches $0.0085$ at $t=0.9$ seconds.  The decrease in $\Theta$ using the MFE reaches up to $99\%$ compared to the baseline. Thus, Fig. \ref{theta} shows the effectiveness of our proposed scheme in limiting the spread of misinformation.

Fig. \ref{qoI} shows the evolution of the QoI over time for both the baseline and the MFE for degree values $k=15$ and $20$, respectively. First when $k=15$ and using the baseline, the QoI will be $16$ at $t=0$. Then, the QoI decreases until it reaches $-12.75$ at $t=0.9$ sec. The decrease in the QoI is due to the increase in the probability $\Theta$ as shown in Fig. \ref{theta}.  When $k=15$ and using the MFE, the QoI is $6$ at $t=0$, since initially all nodes are susceptible and the acceptance probability is zero as shown in Fig. \ref{acceptance}. Then, the QoI decreases with time until it reaches $1.628$ at $t=0.9$ seconds. The decrease in the QoI is mainly due to the increase in the probability $\Theta$ with time, as shown in Fig. \ref{theta}.  When $k=20$ , the QoI resulting from the baseline is $21$ at $t=0$. Then, the QoI decreases with time until it reaches $-17$ at $t=0.9$. When $k=20$ , the QoI resulting from the proposed MFE is $9.8$ at $t=0$. As $t$ increases, the QoI decreases until it reaches $3.64$. The decrease in the QoI when $k=20$ for the baseline and the MFE is due to the same aformentioned reasons for the case when $k=15$.
Thus, the proposed MFE approach achieves a 1.2-fold increase in the value of the QoI compared to the value of the baseline, at $t=0.9$ and when $k=20$.


\begin{figure}[t]{
	\centering
	\includegraphics[width=9 cm,height=5.3cm,angle=0]{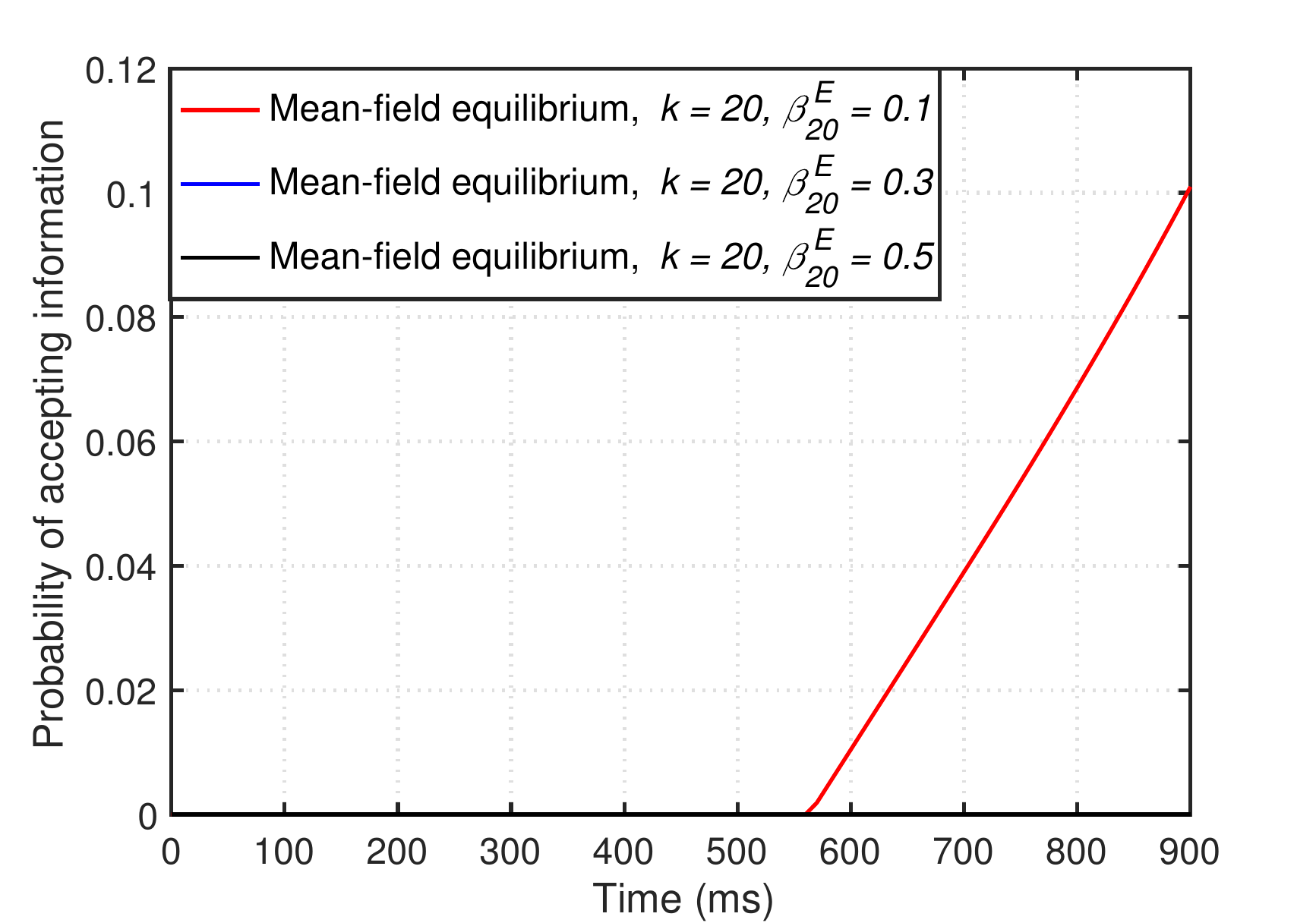}
	\caption{Time evolution of the acceptance probability of nodes with degree $20$  for different values of $\beta^E_{20}$.}\label{alphabeta}
	}\vspace{-0.5 cm}
\end{figure}

Fig. \ref{alphabeta} shows the MFE probability of accepting information by IoBT nodes with degree $20$ over time for $k=20$ and when the value of the $\beta^E_{20}$ is $0.1$, $0.3$, and $0.5$. When $\beta^E_{20}=0.1$,  the acceptance probability is zero for $t \leq 0.56$, respectively. Then, the acceptance probability increases with time until it reaches $0.1$ at $t=0.9$ seconds. When the value of $\beta^E_{20}=0.3,0.5$, the acceptance probability is zero for the entire duration. Thus, Fig. $\ref{alphabeta}$ shows that the acceptance probability decreases as the node has higher capability to identify the misinformation.

\begin{figure}[t]{
	\centering
	\includegraphics[width=9 cm,height=5.3cm,angle=0]{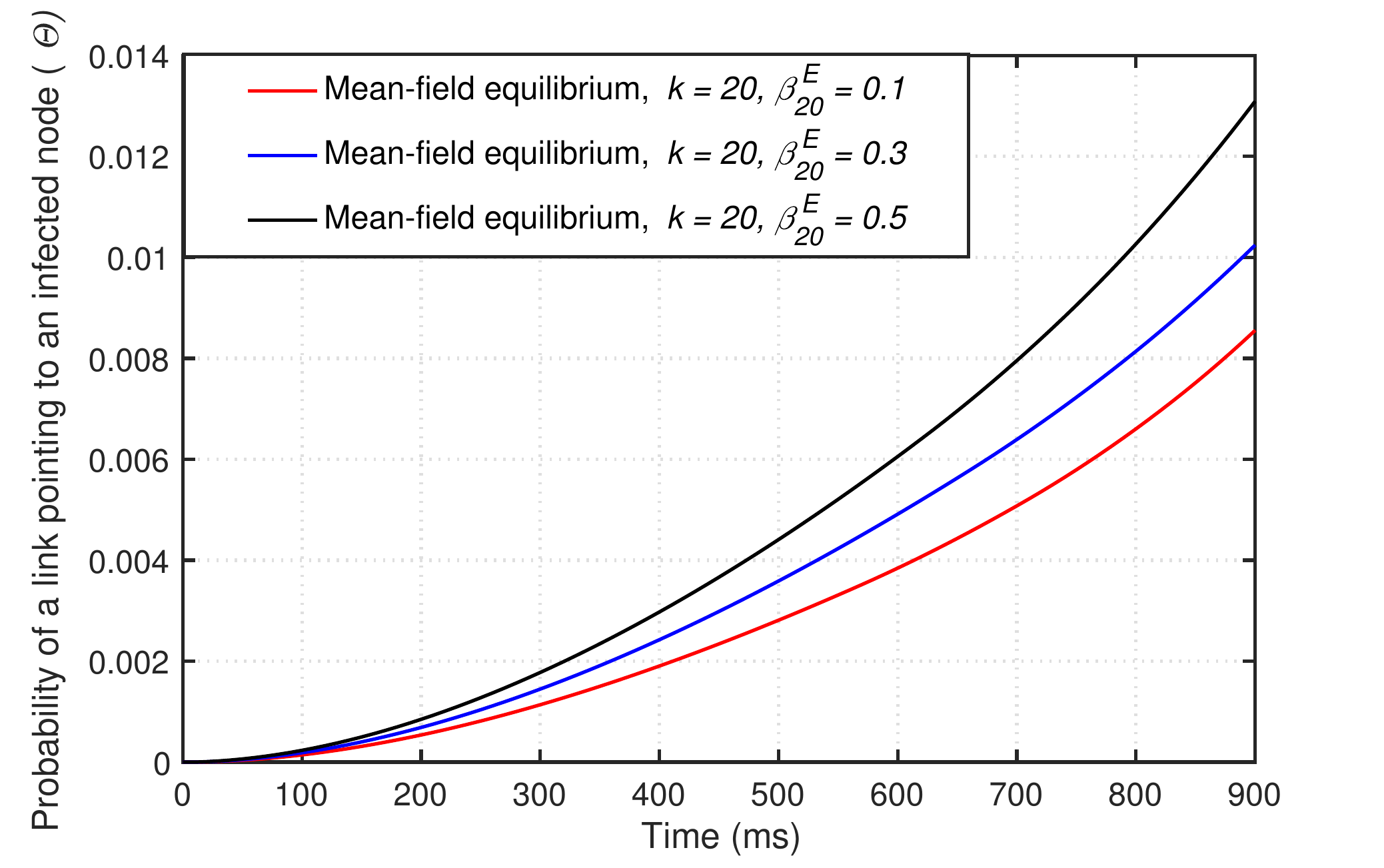}
	\caption{Time evolution of the probability $\Theta(t)$  for different values of $\beta^E_{20}$.}\label{thetabeta}
	}\vspace{-0.3 cm}
\end{figure}

\begin{figure}[t]{
	\centering
	\includegraphics[width=9 cm,height=5.3cm,angle=0]{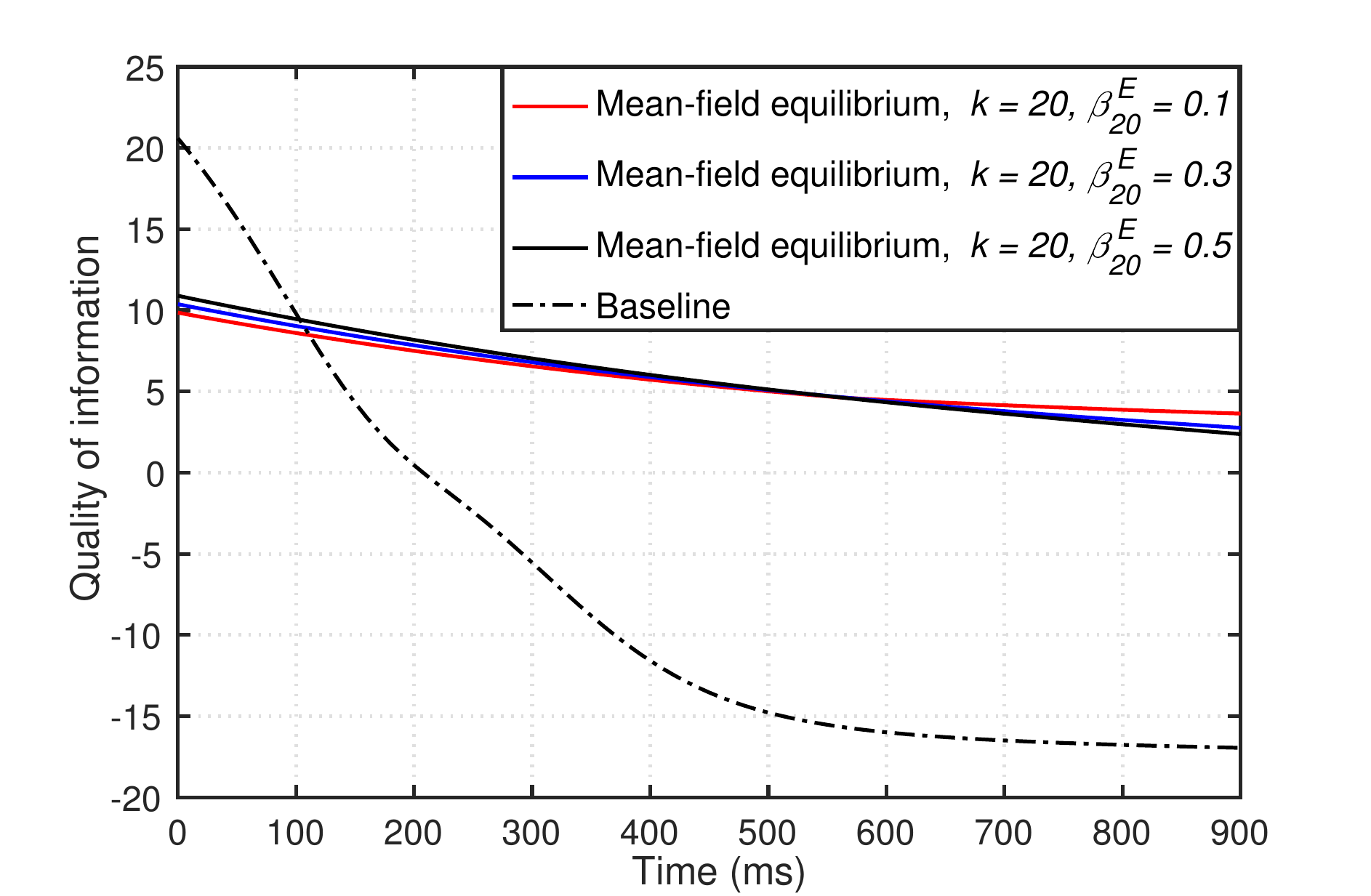}
	\caption{Time evolution of the QoI of nodes with degree $20$  for different values of $\beta^E_{20}$.}\label{qoIbeta}
	}\vspace{-0.5 cm}
\end{figure}

Fig. \ref{thetabeta} shows the probability $\Theta$ of a link pointing to an infected node over time for $k=20$ and when the value of the probability $\beta^E_{20}$ is $0.1$, $0.3$, and $0.5$, respectively. For the three considered values of $\beta^E_{20}$, the probability $\Theta$ increases with time. From Fig. \ref{thetabeta}, we can see that, when $\beta^E_{20}=0.1$, the probability $\Theta$ increases until it reaches $0.0085$ at $t=0.9$. Meanwhile, for $\beta^E_{20}=0.3$ and $0.5$, the probability $\Theta$ reaches $0.01$ and $0.013$, respectively at $t= 0.9$ seconds.
As for the baseline, the probability $\Theta$ is not affected by $\beta^E_4$ since the nodes do not reach the $\emph{E}$ or $\emph{L}$ states. Thus, $\Theta(t)$ is the same as the one shown in Fig. \ref{theta}.
Thus, the MFE shows a considerable decrease in $\Theta$ even in the case when the nodes having degree $k=20$ cannot identify the misinformation (i.e. when $\beta^E_{20}=0.5$). In this case, the decrease in $\Theta$ is $97\%$ compared to the baseline. Further, Fig.  \ref{thetabeta} clearly demonstrates that the spread of misinformation becomes more limited when the IoBT nodes have a higher capability to identify misinformation.

In Fig. \ref{qoIbeta}, we plot the evolution of the QoI over time, for $k=20$ for the three values of $\beta^E_{20}$: $0.1$, $0.3$, and $0.5$.  Fig. \ref{qoIbeta} shows that, for $\beta^E_{20}=0.1$, the QoI is the same as in Fig. \ref{qoI}. However, when $\beta^E_{20}=0.3$, the QoI decreases with time until it reaches $2.7$ at $t=0.9$ seconds. When $\beta^E_{20}=0.5$, the QoI decreases with time until it reaches $2.37$ at $t=0.9$ seconds. 
Using the baseline, the QoI  is not affected by $\beta^E_4$. Thus, as demonstrated earlier in Fig. \ref{theta}, the QoI decreases until it reaches $-17$ at $t=0.9$. Fig. \ref{qoIbeta} further shows that for $t \leq 0.57$, the QoI increases with $\beta^E_{20}$ due to the low value of the probability $\Theta$. Thus, a higher value of $\beta^E_{20}$ will result in a higher QoI. For $t \leq 0.57$, the QoI decreases with $\beta^E_{20}$ since, in the considered time duration, $\Theta$ increases at a faster rate as shown in Fig. \ref{thetabeta}.
Also, Fig. \ref{qoIbeta} shows the improvement in the QoI compared to the baseline even when the nodes of degree $k=20$ are not able to identify the misinformation. In this case, the improvement in the QoI reaches up to $99\%$ when $\beta^E_{20}=0.1$ and at $t=0.9$ sec. 

\begin{figure}[t]{
	\centering
	\includegraphics[width=9 cm,height=5.3cm,angle=0]{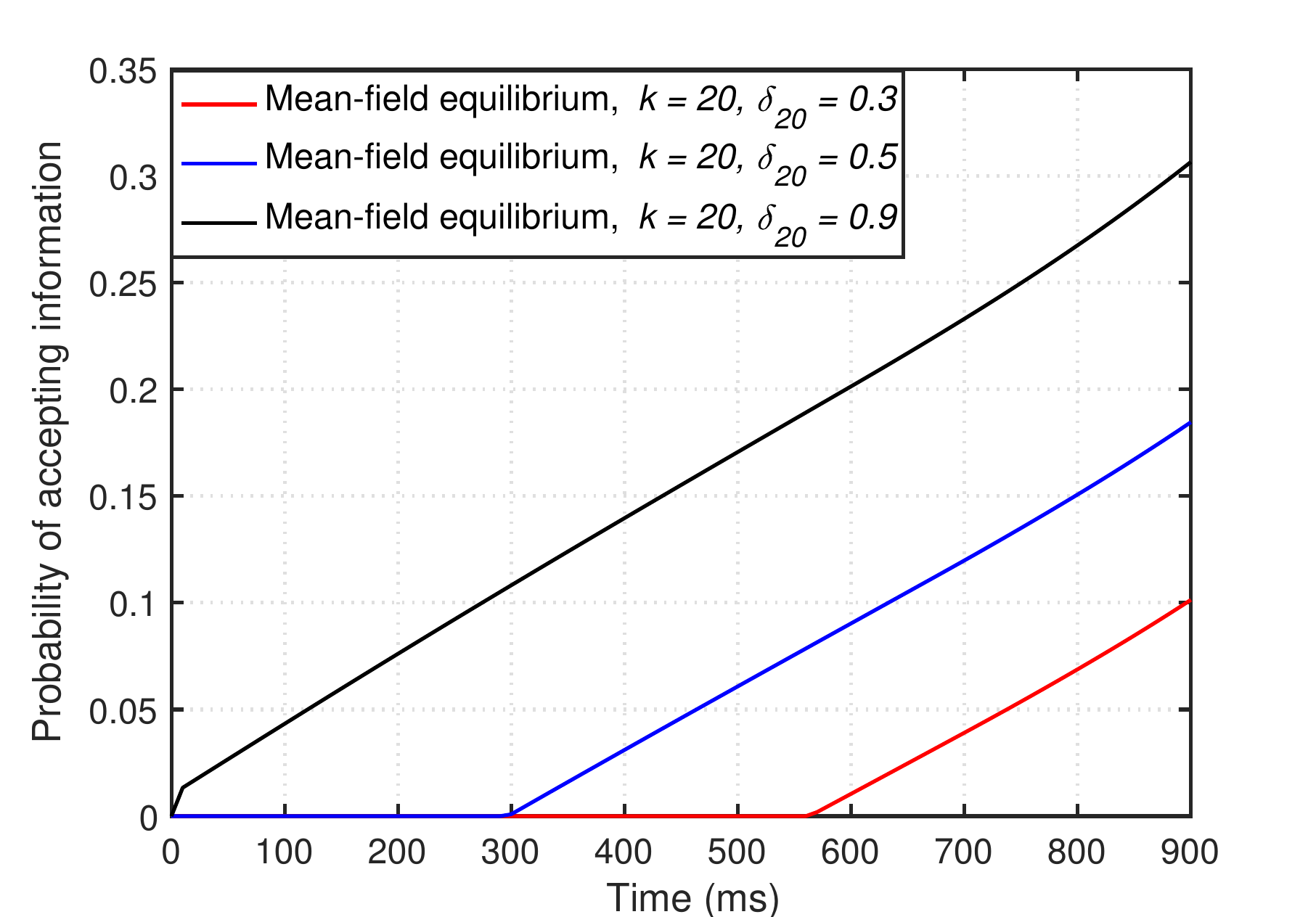}
	\caption{Time evolution of the acceptance probability of nodes with degree $20$ for different values of 
$\delta_4$.}\label{alphadelta}
	}\vspace{-1 cm}
\end{figure}

Fig. \ref{alphadelta} shows the MFE probability of accepting information by IoBT nodes with degree $20$ over time for $k=20$ and when the value of the delay $\delta_{20}$ is $0.3$, $0.5$, and $0.9$. When $\delta_{20}=0.3, 0.5$,  the acceptance probability is zero for $t \leq 0.3$ and  $t \leq 0.56$, respectively. Then, the acceptance probability increases with time until it reaches $0.1$ and $0.1846$, respectively. When the value of $\delta_{20}=0.9$, the acceptance probability is zero at $t=0$. Then, it increases with time until it reaches $0.3$ at $t=0.9$. Fig.  \ref{alphadelta} shows that, at MFE, the acceptance probability increases with the delay $\delta_{20}$ in order to prevent the QoI from deteriorating due to the age of information.

\begin{figure}[t]{
	\centering
	\includegraphics[width=9 cm,height=5.3cm,angle=0]{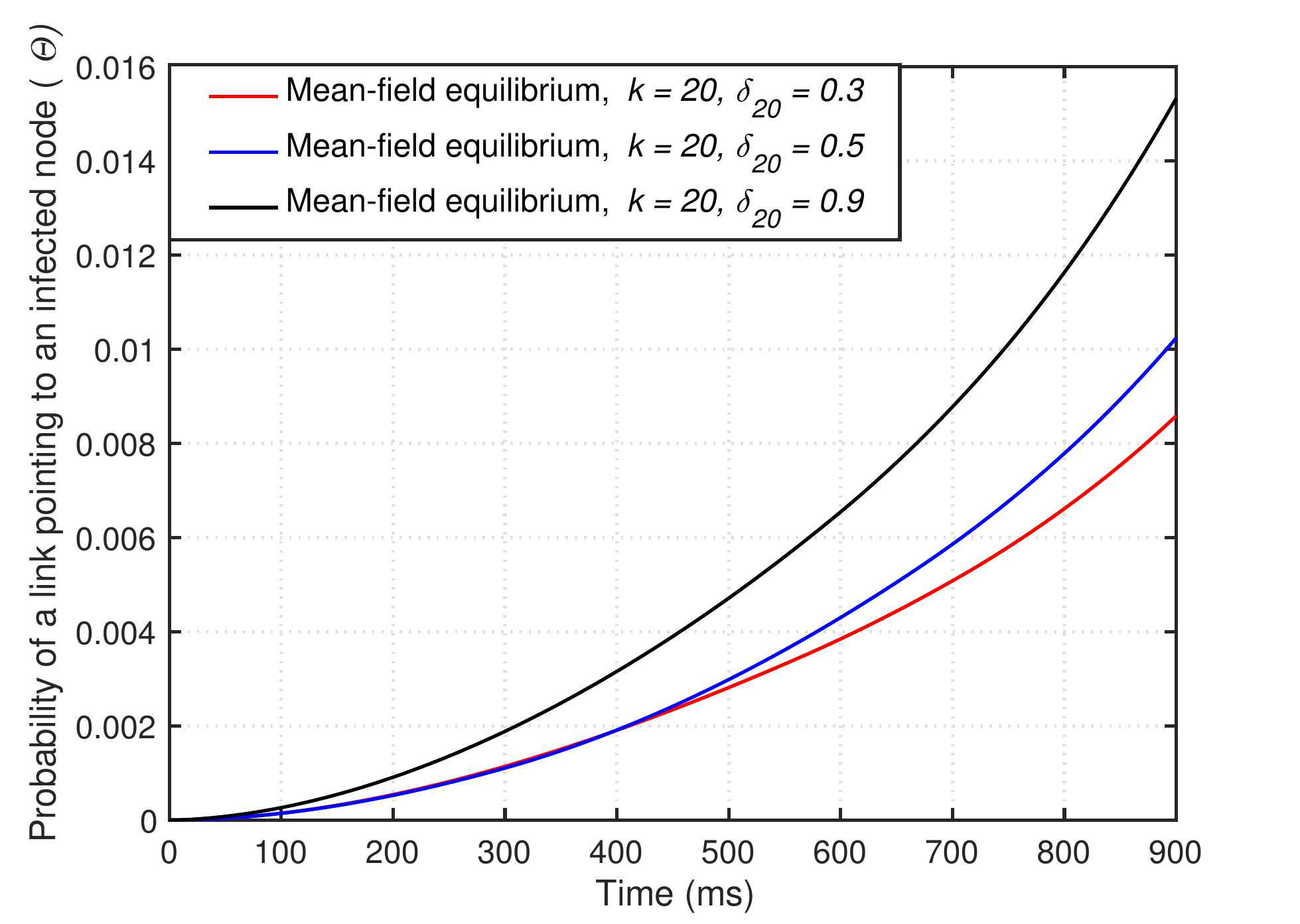}
	\caption{Time evolution of $\Theta(t)$ of nodes with degree $20$ for different values of 
$\delta_{20}$.}\label{thetadelta}
	}\vspace{-0.5 cm}
\end{figure}

\begin{figure}[t]{
	\centering
	\includegraphics[width=9 cm,height=5.3cm,angle=0]{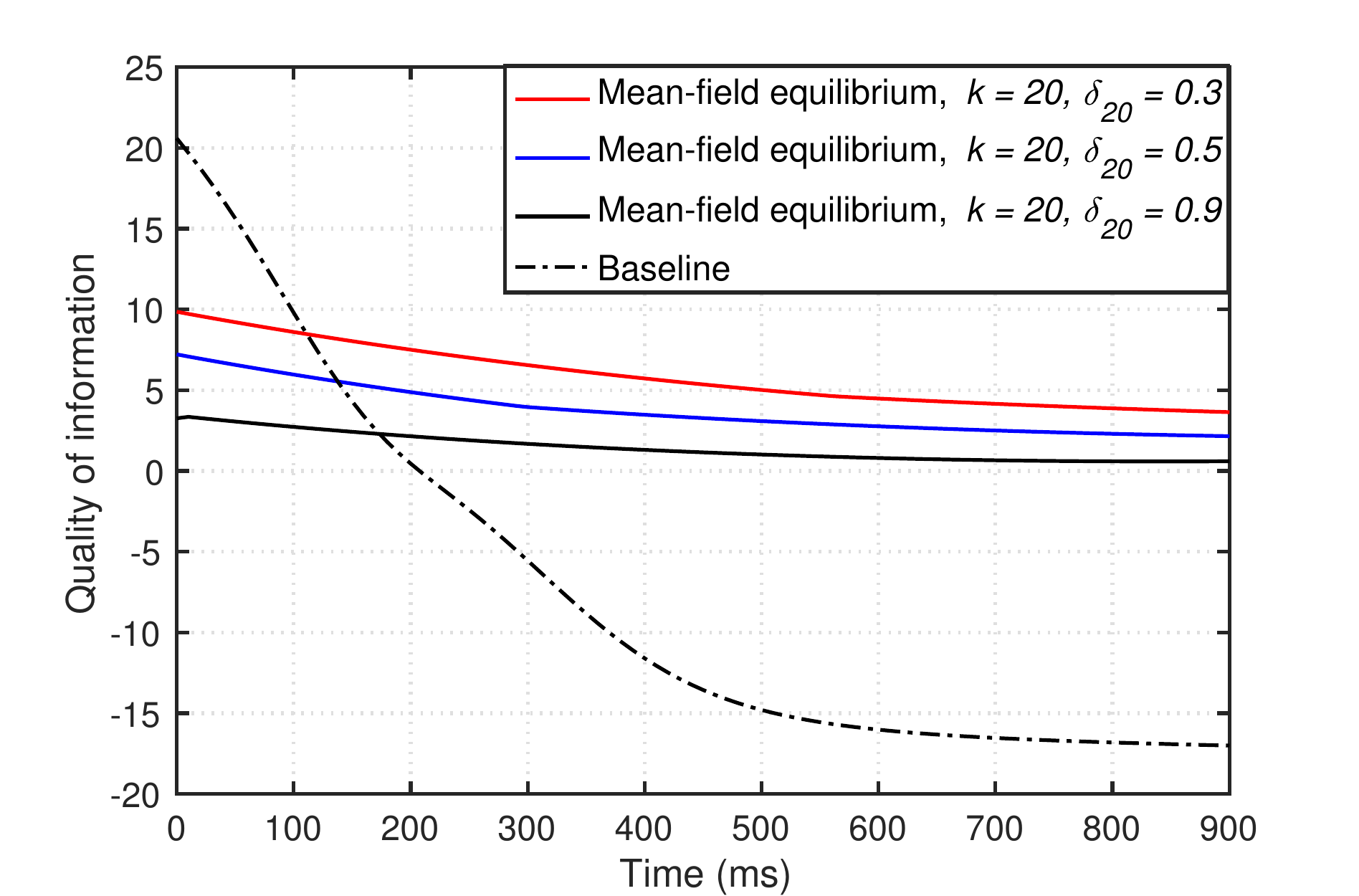}
	\caption{Time evolution of QoI of nodes with degree $20$ for different values of 
$\delta_{20}$.}\label{qoIdelta}
	}\vspace{-0.4 cm}
\end{figure}

Fig. \ref{thetadelta} shows the probability $\Theta$ of a link pointing to an infected node versus time for $k=20$ and when the value of the delay $\delta_{20}$ is $0.3$, $0.5$, and $0.9$.  For the three considered values of $\delta_{20}$, Fig. \ref{thetadelta} shows that the probability $\Theta$ increases with time. When $\delta_{20}=0.3$ and $t=0$, the probability $\Theta$ is $0$. Then, it increases with time until it reaches $0.0085$ at $t=0.9$ seconds. When $\delta_4=0.5$, the probability $\Theta$ increases with time until it reaches $0.01$ at $t=0.9$ . Similarly, when $\delta_{20}=0.9$, the probability $\Theta$ reaches $0.015$ at $t=0.9$. 
Also, as shown in Fig. \ref{thetadelta} the probability $\Theta$ increases with an increase in $\delta_{20}$. The increase in $\Theta$ is due to the fact that the acceptance probability increases with $\delta_{20}$, as shown in Fig. \ref{alphadelta}.
For the baseline, the probability $\Theta$ is not affected with change in $\delta_{20}$ since the nodes do not enter the $\emph{E}$ or $\emph{L}$ states. Thus, $\Theta(t)$ is the same as the one shown in Fig. \ref{theta}. Thus, Fig. \ref{thetadelta} confirms that the proposed MFE approach can achieve a significant decrease in the probability $\Theta$ compared to the baseline, reaching up to $99\%$ when $\delta_{20}=0.3$ and at $t=0.9$ seconds.

\vspace{-0.1 cm}
Fig. \ref{qoIdelta} shows how the QoI resulting from the MFE and the baseline will vary over time, for $k=20$, and for $\delta_{20} =0.3$, $0.5$, and $0.9$ for both the MFE and the baseline. When $\delta_4=0.3$ and for the MFE, the QoI is $10$ at $t=0$. Then, as time increases, the QoI decreases until it reaches $3.64$ at $t=0.9$ seconds. When $\delta_{20}=0.5,0.9$, the evolution of QoI is similar to the case when $\delta_{20}=0.3$ and the minium value of QoI is $2.15$ and $0.6$, respectively at $t=0.9$. Using the baseline, the QoI is not affected with $\delta_{20}$. Hence as demonstrated earlier in Figs. \ref{qoI} and \ref{qoIbeta}, the QoI decreases with time until it reaches $-17$ at $t=0.9$.
Thus, Fig. \ref{qoIdelta} shows that the QoI deteriorates with an increase in the information processing delay, due to the increase in the probability $\Theta$ as shown in Fig. \ref{thetadelta}. Nonetheless, the MFE maintains a significant gain in the QoI compared to the baseline even with high information processing delays. In  particular, the MFE achieves a 1.2-fold increase in the quality of information when $\delta_4=0.3$ and $t=0.9$ compared to the baseline.
\begin{figure}[t]{
	\centering
	\includegraphics[width=9 cm,height=5.3cm,angle=0]{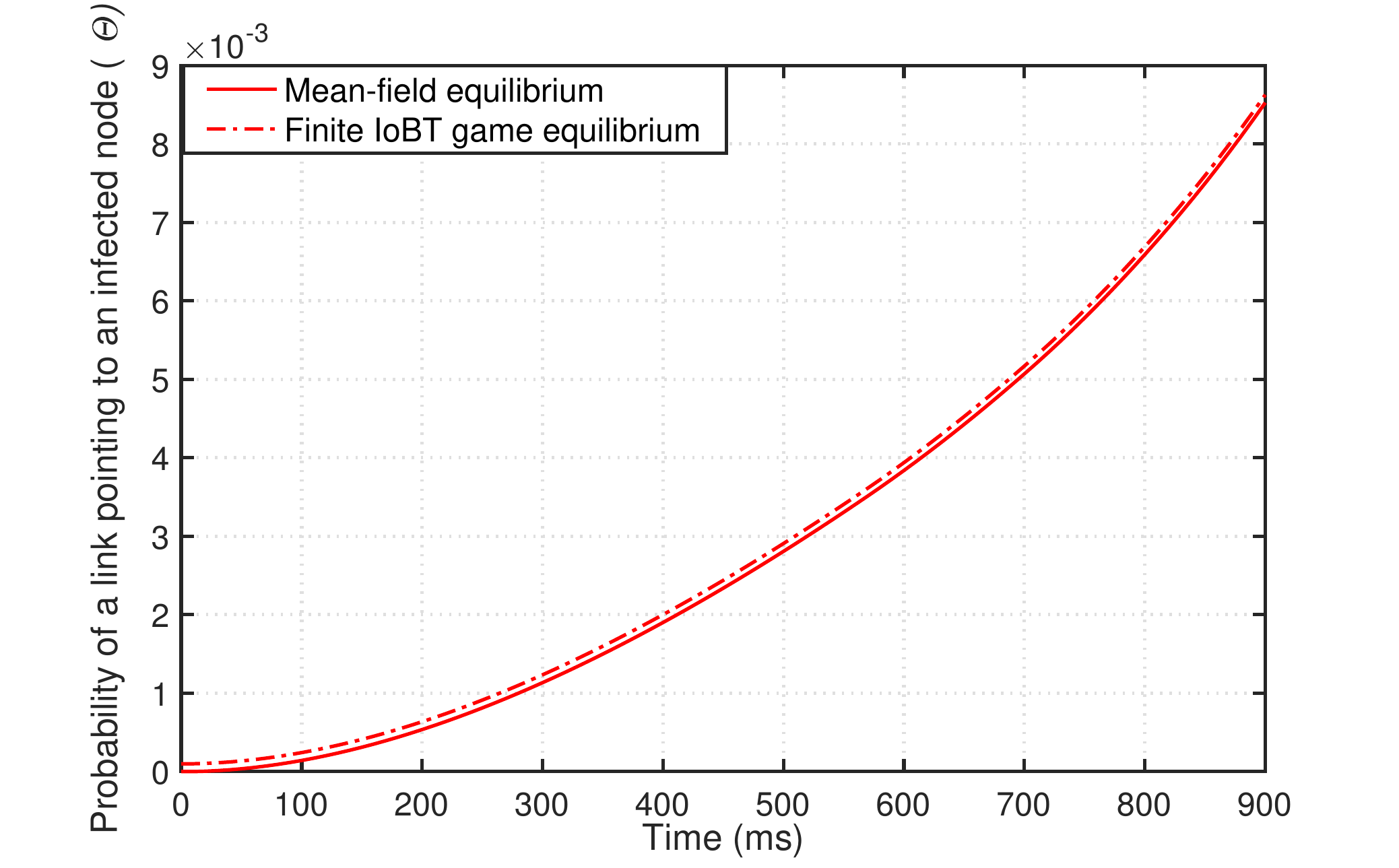}
	\caption{Time evolution of the probability $\Theta(t)$ for the IoBT mean-field game and the finite IoBT game.}\label{thetasim}
	}\vspace{-1 cm}
\end{figure}

\begin{figure}[t]{
	\centering
	\includegraphics[width=9 cm,height=5.3cm,angle=0]{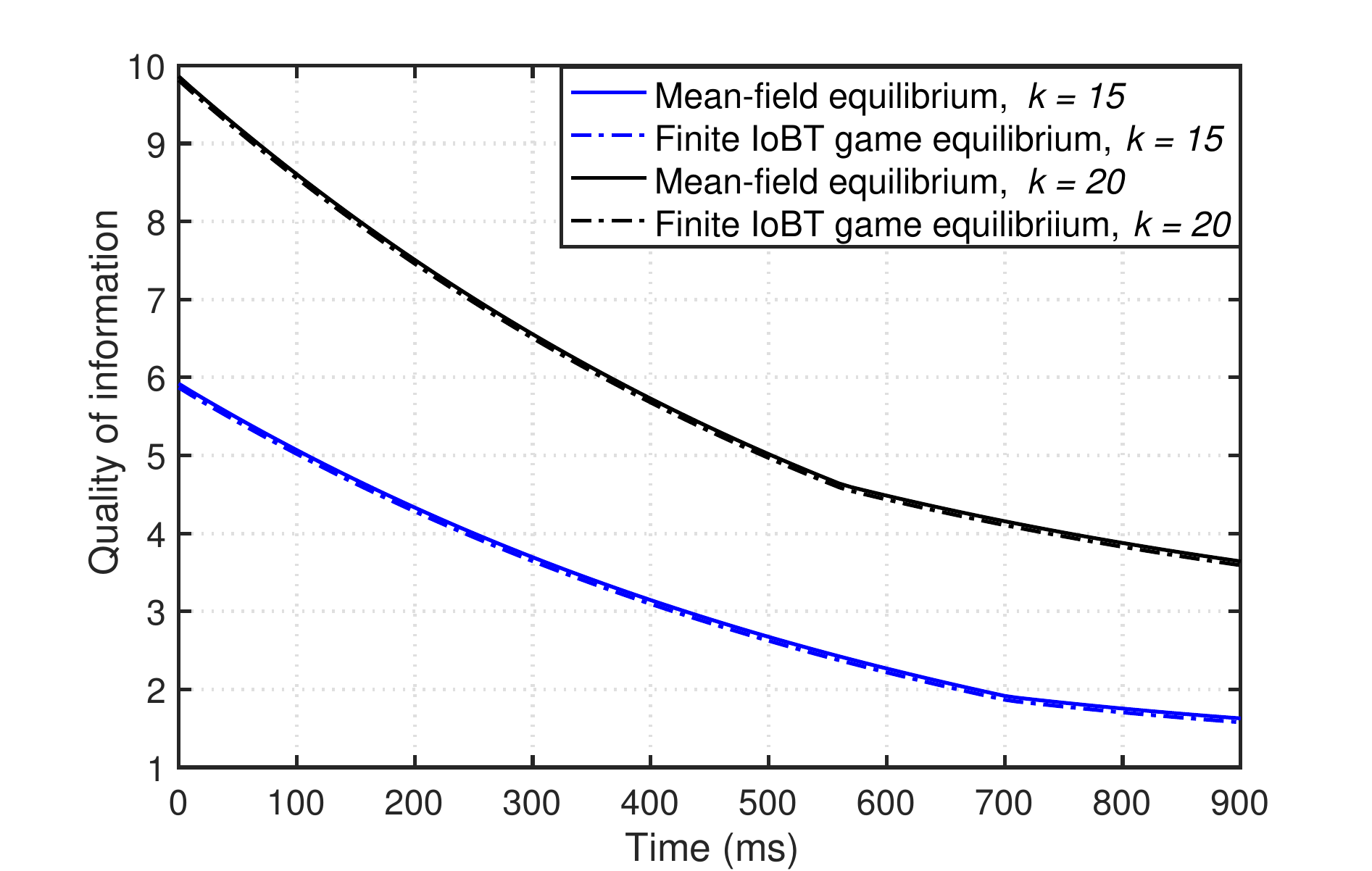}
	\caption{Time evolution of the QoI for the IoBT mean-field game and the finite IoBT game.}\label{qoIsim}
	}
\end{figure}

\vspace{-0.1 cm}
Next, for the considered simulation values, we consider for the finite IoBT case when $N=10,000$ and compute, using Algorithm 1 tailored to the finite game, the values of the probability $\Theta$ and the QoI at equilibrium shown in Figs. \ref{thetasim} and \ref{qoIsim}. Figs. \ref{thetasim} and \ref{qoIsim} shows that the values of $\Theta$ and the QoI of the finite IoBT game coincide with the values of the mean-field game, which confirms the convergence of the finite game to the mean-field game for large $N$.

\vspace{-0.4 cm}
\section{Conclusion}
In this paper,  we have considered the problem of misinformation propagation in an IoBT in which the nodes seek to determine the optimal probability of accepting the information. We have formulated the problem as a finite state mean-field game with multiclass agents. We have proposed an algorithm based on the forward backward sweep algorithm to find the mean-field equilibrium. We have analyzed the finite IoBT game and derived the conditions of convergence of the finite IoBT game to the mean-field game as the number of nodes tends to infinity. Our results have shown our proposed scheme can achieve a $1.2$-fold increase in the QoI compared to the value of the baseline when the nodes are transmitting. Further, our proposed scheme can reduce the proportion of infected nodes by $99\%$ compared to the baseline.

\vspace{-0.3 cm}

\appendix
\vspace{-0.3 cm}
\section*{Appendix A: Proof of Proposition 1}

Each characteristic, presented in Proposition 1, of the studied IoBT game is proven as follows:

\begin{enumerate}
\item The transitional rate $G^{jl}_{ik}(\alpha_{ik}(t), \Theta(t))$  is only a function of $\alpha_{ik}(t)$ when $j=\emph{S}$, and, in this case, it is a linear function of $\alpha_{ik}(t)$ and therefore Lipschitz in $\alpha_{ik}(t)$.

\item Proving that the best response $\alpha^*_{ik}(\Delta_l \boldsymbol{u}_{ik}, \boldsymbol{\gamma}(t))$ is Lipschitz in $\Delta_l \boldsymbol{u}_{ik}$, $\Theta(t)$, $\eta(t)$ can be shown using a similar proof as \cite[Proposition 1]{finitemean-field} and using the fact the  transitional rate $G^{jl}_{ik}(\alpha_{ik}(t), \Theta(t))$ is Lipchitz in $\alpha_{ik}(t)$.
One direct consequence is that the best response $\alpha^*_{ik}(\Delta_l \boldsymbol{u}_{ik}, \boldsymbol{\gamma}(t))$ is Lipchitz in $m_{ik}(t),$ $\forall (i,k) \in \mathcal{C}$ since both $\Theta(t)$ and $\eta(t)$ are linear functions of $m_{ik}(t)$ $\forall (i,k) \in \mathcal{C}$.
However, the proof relies on the assumption that that cost $v_{ik}(l,\alpha^l_{ik}(t),\boldsymbol{\gamma}(t))$ is strongly convex w.r.t $\alpha^l_{ik}(t)$. 

By computing the second order partial derivative of $v_{ik}$ with respect to $\alpha_{ik}(t)$,
\vspace{-0.5 cm}

\small
\begin{eqnarray}
&&\hspace{-1 cm}\frac{\partial^2v_{ik}}{\partial \alpha_{ik}(t)^2}(S,\alpha_{ik}(t),\Theta(t))=(L_{ik}(\Theta(t))(k\eta(t)+1)+(k\eta(t)-k\Theta(t)-\lambda_{ik}(t)-(1-\lambda_{ik})k\eta(t))\nonumber\\
&&\hspace{-0.9 cm}-L_{ik}(\Theta(t))(\beta^L_{ik}(t)(k\eta(t)+1-\kappa\delta_{ik}(t))-(\beta^E_{ik}(k\eta(t)-k\Theta(t)-\lambda_{ik}(t)-(1-\lambda_{ik})k\eta(t)\nonumber\\
&&\hspace{-1 cm}-\kappa\delta_{ik}(t)(1-L_{ik}(\Theta(t)))).\nonumber
\end{eqnarray}
\normalsize

The derivative $\frac{\partial^2v_{ik}}{\partial \alpha_{ik}(t)^2}$ is not necessarily positive since the second term (which corresponds to the QoI when the node receives misinformation) could be negative.  Thus, in order to ensure that the cost is strongly convex, the values of $Q_{ik}(\boldsymbol{\gamma}(t))$ are scaled such that the resulting values are always positive.
 Let $Q'_{ik}(\boldsymbol{\gamma}(t))$ the scaled valued. $Q'_{ik}(\boldsymbol{\gamma}(t))$ can be possibly defined as $Q'_{ik}(\boldsymbol{\gamma}(t))=Q_{ik}(\boldsymbol{\gamma}(t))+S_k$ where $S_k$ is the scaling factor and is given by $S_k=k+2$.

By normalizing the QoI values, the second order partial derivative of $v_{ik}$ with respect to $\alpha_{ik}(t)$ becomes
\vspace{-0.5 cm}

\small
\begin{eqnarray}
&&\hspace{-1 cm}\frac{\partial^2v_{ik}}{\partial \alpha_{ik}(t)^2}(S,\alpha_{ik}(t),\Theta(t))=(L_{ik}(\Theta(t))(k\eta(t)+1+S_k)+(k\eta(t)-k\Theta(t)-\lambda_{ik}(t)-(1-\lambda_{ik})k\eta(t)+S_k)\nonumber\\
&&\hspace{-0.5 cm}-L_{ik}(\Theta(t))(\beta^L_{ik}(t)(k\eta(t)+1+S_k-\kappa\delta_{ik}(t))-(\beta^E_{ik}(k\eta(t)-k\Theta(t)+S_k-\lambda_{ik}(t)-(1-\lambda_{ik})k\eta(t)\nonumber\\
&&\hspace{-0.5 cm}-\kappa\delta_{ik}(t)(1-L_{ik}(\Theta(t)))).\nonumber
\end{eqnarray}
\normalsize
It can be easily shown that $\frac{\partial^2v_{ik}}{\partial \alpha_{ik}(t)^2}$ is lower bounded by $\kappa=(1-\beta_{ik}^L)(2k-1)+\beta_E\delta_{ik}>0$ assuming $\beta_L \geq \beta_E$. Hence, the cost in this case is strongly convex with parameter $\kappa$ assuming that either $\delta_{ik} \neq 0$ or $\beta^L_{ik} \neq 1$.


\item In order to prove this property, we first note that the transitional rate is only a function of $\alpha_{ik}(t)$ and $\Theta(t)$ only when the state is  $\emph{S}$.
We consider the transitional rate $G^{SI}_{ik}(\alpha^{S}_{ik}(t), \Theta(t)))=\alpha^{S}_{ik}(t)R_{ik}(\Theta(t))$ and compute its partial derivative with respect to $\Theta(t)$:
\small
\begin{equation}
\frac{\partial G^{SI}_{ik}}{\partial \Theta(t)}(\alpha^*_{ik}(t),\Theta(t))= \frac{\partial}{\partial \Theta(t)}\alpha^*_{ik}(t)\Theta(t)=k\Theta(t)\frac{\partial}{\partial \Theta(t)} \alpha^*_{ik}(t)+k\alpha^*_{ik}(t).
\end{equation}
\normalsize
The partial derivative $\frac{\partial}{\partial \Theta(t)}\alpha^*_{ik}(t)$ is bounded  since $\alpha_{ik}(t)$ is Lipschitz in  $\Theta(t)$  according to 2). Further, $\Theta(t)$ and $\alpha_{ik}(t)$ are bounded by $1$. Thus, $\frac{\partial G^{S_I}_{ik}}{\partial \Theta(t)}(\alpha^*_{ik}(t))$ is bounded, and the transitional rate $ G^{SI}_{ik}$ is Lipschitz in $\Theta(t)$.

Further, $G^{SI}_{ik} (\alpha_{ik}(t), \Theta(t))$ is a linear function of $\alpha_{ik}(t)$ and therefore is Lipschitz in $\Delta_S u$ since $\alpha_{ik}(t)$ is Lipschitz in $\Delta_S u$. This property can be proved for the remaining transitional probabilities using a similar method.
\item The immediate cost $v_{ik}$ and its derivative $\nabla_\alpha v_{ik}$ can be similarly proven to be Lipchitz in $\Theta$ and $\eta$ by showing that the partial derivatives with respect to $\Theta$ and $\eta$ are bounded. 
\item This property easily follows from 2), 3), and 4).
\end{enumerate}
\vspace{-0.3 cm}
\section*{Appendix B: Proof of Lemma 2}

Let $W^N_{ik}(l,t)=\mathbb{E}\Big[(\boldsymbol{u}^l_{ik}(t)-\boldsymbol{u}^{N,\boldsymbol{n},l}_{ik}(t))^2\Big]$. Thus, $W^N_{ik}(t)=\max_{l \in \mathcal{S}}W^N_{ik}(l,t)$.
To prove the lemma,  we apply Dynkin formula on functions of the process $(l,\boldsymbol{n}_{ik})$ . First, we define the infinitesimal generator acting on a function of the process $(l,\boldsymbol{n}_{ik})$ $\varphi: (\mathcal{S}, \mathcal{N}^{\mathcal{S}}, [0,T]) \rightarrow \mathbb{R}$ as
\small
\begin{eqnarray}
&&\hspace{-0.5 cm}A_{ik} \varphi(l, \boldsymbol{n}_{ik}, s)=\sum_{j \in \mathcal{S}} G^{ik}_{lj}(\alpha^{N,l}_{ik}(s))[\varphi(j, \boldsymbol{n}_{ik}(s), s)-\varphi(l, \boldsymbol{n}_{ik}(s), s)]\nonumber\\
&&\hspace{3.7 cm}+\sum_{j \in \mathcal{S}}\sum_{y \in \mathcal{S}}n^y_{ik} G^{N,ik}_{yj}(\alpha^{N,y}_{ik}(s))[\varphi(l, \boldsymbol{n}_{ik}(s)+\boldsymbol{e}_{jy}, s)-\varphi(l, \boldsymbol{n}_{ik}(s), s)], \label{genconv}
\end{eqnarray}
\normalsize
where $\alpha^{N,y}_{ik}=\alpha^{N,y*}_{ik}(\boldsymbol{\gamma}_N(\boldsymbol{n}(t)+\boldsymbol{e}^{ik}_{ly}),\Delta_y u^{N,\boldsymbol{n}(t)+\boldsymbol{e}^{ik}_{ly}}_{ik})$ for $(i,k) \neq (i',k')$ ($(i',k')$ is the class of the reference player) and $\alpha^{N,y}_{ik}=\alpha^{N,y*}_{ik}(\boldsymbol{\gamma}_N(\boldsymbol{n}(t)-\boldsymbol{e}^{ik}_{y}),\Delta_y u^{N,\boldsymbol{n}(t)-\boldsymbol{e}^{ik}_{y}}_{ik})$ for $(i,k)\neq (i',k')$ are the equilibrium acceptance probabilities for the finite IoBT game. Using Dynkin formula, we have
\small
\begin{equation}
\mathbb{E}[\varphi(l_{ik}(T),\boldsymbol{n}_{ik}(T),T)-\varphi(l_{ik}(t),\boldsymbol{n}_{ik}(t),t)]
= \mathbb{E} \bigg[\int_{t}^T \frac{d \varphi}{dt}(l_{ik}(s),\boldsymbol{n}_{ik}(s),s) + A_{ik} \varphi(l_{ik}(s), \boldsymbol{n}_{ik}(s), s) ds\bigg],\label{dynkin}
\end{equation}
\normalsize
where $l_{ik}(s)$ is the state of the reference player at time $s$.

Next, we define $\varphi_l(j, \boldsymbol{n}_{ik}(t), t)=(\boldsymbol{u}^l_{ik}(t)-\boldsymbol{u}^{N,\boldsymbol{n},l}_{ik}(t))^2$. Using (\ref{dynkin}), we have
\small
\begin{eqnarray}
&&\hspace{-0.6 cm}W^N_{ik}(l,t)-W^N_{ik}(l,T)=-\mathbb{E}\Big[(\boldsymbol{u}^{N,\boldsymbol{n},l}_{ik}(t)-\boldsymbol{u}^l_{ik}(t))^2\Big]+\mathbb{E}\Big[(\boldsymbol{u}^{N,\boldsymbol{n},l}_{ik}(T)-\boldsymbol{u}^l_{ik}(T))^2\Big]\nonumber
\end{eqnarray}
\begin{eqnarray}
&&\hspace{-0.6 cm}=\mathbb{E}\int_{t}^T 2(\boldsymbol{u}^{N,\boldsymbol{n},l}_{ik}(s)-\boldsymbol{u}^l_{ik}(s))\frac{d}{ds}(\boldsymbol{u}^{N,\boldsymbol{n},l}_{ik}(s)-\boldsymbol{u}^l_{ik}(s))ds + \int_{t}^T \sum_{jy}n^y_{ik} G^{N,ik}_{yj}(\alpha^{N,y}_{ik}(s))[\varphi(l, \boldsymbol{n}_{ik}(s)+\boldsymbol{e}^{ik}_{jy}, s)-\varphi(l, \boldsymbol{n}_{ik}(s), s)] \nonumber
\end{eqnarray}
\begin{eqnarray}
&&\hspace{-0.6 cm}=\mathbb{E}\int_{t}^T 2(\boldsymbol{u}^{N,\boldsymbol{n},l}_{ik}(s)-\boldsymbol{u}^l_{ik}(s))\Big(\sum_{y,j}\eta^l_{ik}(y,j,\boldsymbol{n})(u^{N,\boldsymbol{n}+\boldsymbol{e}^{ik}_{jy},l}_{ik}(s) - u^{N,\boldsymbol{n},l}_{ik}(s)) - h(\Delta_l u^{N,\boldsymbol{n}}_{ik} ,\boldsymbol{m}^N(s), l)+h(\Delta_l \boldsymbol{u}_{ik}, \boldsymbol{m}(s), l)ds, \nonumber\\
&&\hspace{-0.6 cm}+\mathbb{E} \int_{t}^T \sum_{jy}n^y_{ik} G^{N,ik}_{yj}(\alpha^{N,y}_{ik}(s))(u^{N,\boldsymbol{n}+\boldsymbol{e}^{ik}_{jy},l}_{ik}(s) - u^{N,\boldsymbol{n},l}_{ik}(s))^2 -(u^{N,\boldsymbol{n},l}_{ik}(s) - u^{N,\boldsymbol{n},l}_{ik}(s))^2), \nonumber
\end{eqnarray}
\begin{eqnarray}
&&\hspace{-0.6 cm}= \mathbb{E} \int_{t}^T \sum_{jy}n^y_{ik} G^{N,ik}_{yj}(\alpha^{N,y}_{ik}(s))(u^{N,\boldsymbol{n}+\boldsymbol{e}^{ik}_{jy},l}_{ik}(s) - u^{N,\boldsymbol{n},l}_{ik}(s))^2 ds 
+ \mathbb{E}\int_{t}^T (2(u^{N,\boldsymbol{n},l}_{ik}(s) - u^{N,\boldsymbol{n},l}_{ik}(s)) (h(\Delta_l \boldsymbol{u}_{ik}, \boldsymbol{m}(s), l),\nonumber\\
&&\hspace{-0.3 cm}- h(\Delta_l u^{N,\boldsymbol{n}}_{ik} ,\boldsymbol{m}^N(s),l)ds.
\end{eqnarray}
\normalsize

From Remark \ref{grad}, we have  $ \sum_{jy}n^y_{ik} G^{N,ik}_{yj}(\alpha^{N,y}_{ik}(s))(u^{N,\boldsymbol{n}+\boldsymbol{e}_{jy},l}_{ik}(s) - u^{N,\boldsymbol{n},l}_{ik}(s))^2 \leq \frac{K_2}{N}$. Then, since the terminal conditions are zero, we have
\small
\begin{equation}
W^N_{ik}(t)\leq \frac{K_3}{N}+2\mathbb{E}\int_{t}^T(u^{N,\boldsymbol{n},l}_{ik}(s) - u^{N,\boldsymbol{n},l}_{ik}(s)) (h(\Delta_l \boldsymbol{u}_{ik}, \boldsymbol{m}(s), l)- h(\Delta_l u^{N,\boldsymbol{n}}_{ik} ,\boldsymbol{m}^N(s))ds, \label{fiftythree}
\end{equation}
\normalsize
where $K_3=K_2T$. Using Proposition 1,  $h$ is Lipschitz function of $\Delta_l \boldsymbol{u}_{ik}$ and $\boldsymbol{m}_{ik}(t)$ $\forall (i,k) \in \mathcal{C}$. Hence,
\small
\vspace{-0.3 cm}
\begin{equation}
(h(\Delta_l \boldsymbol{u}_{ik}, \boldsymbol{m}(s), l)- h(\Delta_l u^{N,\boldsymbol{n}}_{ik}, \boldsymbol{m}^N(s),l)) \leq K_4( \sum_{(r,v) \in \mathcal{C}}||\frac{\boldsymbol{n_{rv}}(s)}{N_{rv}}-\boldsymbol{m_{rv}}(s)||+||u^{N,\boldsymbol{n}}_{ik}-u_{ik}||).\label{fiftytwo}
\end{equation}
\normalsize
Then, from (\ref{fiftythree}) and (\ref{fiftytwo}) and using the property $ab<a^2+b^2$, we have
\small
\begin{eqnarray}
W^N_{ik}(t)&\leq& \frac{K_3}{N}+K_4 \mathbb{E} \int_{t}^T\sum_{(r,v) \in \mathcal{C}}||\frac{\boldsymbol{n_{rv}}(s)}{N_{rv}}-\boldsymbol{m_{rv}}(s)||^2+||u^{N,\boldsymbol{n}}_{ik}(s)-u^{ik}(s)||^2ds, \nonumber
\end{eqnarray}
\begin{eqnarray}
W^N_{ik}(t)&\leq& \frac{K_3}{N}+K_4 \mathbb{E} \int_{t}^T W^N_{ik}(s) + \sum_{(r,v) \in \mathcal{C}}V^N_{rv}(s)ds,\nonumber
\end{eqnarray}
\begin{eqnarray}
&\leq& \frac{K_3}{N}+K_4 \mathbb{E} \int_{t}^T W^N_{ik}(s) + \sum_{(r,v) \in \mathcal{C}}V^N_{rv}(s) ds,\nonumber
\end{eqnarray}
\begin{eqnarray}
&\leq& \frac{C_1}{N}+C_1 \mathbb{E} \int_{t}^T W^N_{ik}(s) + \sum_{(r,v) \in \mathcal{C}}V^N_{rv}(s) ds,\nonumber
\end{eqnarray}
\normalsize
where $C_1=\max\{K_3,K_4\}$.
\normalsize
\vspace{-1 cm}
\section*{Appendix C: Proof of Lemma 3}

By applying Dynkin's Formula  (\ref{dynkin}) with $\varphi_l(j,\boldsymbol{n}_{ik},t)=(\boldsymbol{m}^l_{ik}(t)-\frac{\boldsymbol{n}^l_{ik}(t)}{N_{ik}})^2$ for all $l \in \mathcal{S}$, we get
\small
\begin{eqnarray}
V^N_{ik}(l,t)-\frac{(\nu^l_{ik})(1-\nu^l_{ik})}{2}=\mathbb{E}\int_{0}^t \frac{d \varphi_l}{dt}(l_{ik}(s),\boldsymbol{n}_{ik}(s),s) + A_{ik} \varphi^l(l_{ik}(s), \boldsymbol{n}_{ik}(s), s) ds,
\end{eqnarray}
\normalsize
where 
\small
\vspace{-0.3 cm}
\begin{equation}
\frac{d \varphi_l}{dt}(l_{ik}(s),\boldsymbol{n}_{ik}(s),s)=-2\Big(\frac{\boldsymbol{n}^l_{ik}(s)}{N_{ik}}-\boldsymbol{m}^l_{ik}(s)\Big)\sum_{j \in \mathcal{S}}G^{lj}_{ik}(\alpha^l_{ik}(s))m^j_{ik}(s).
\end{equation}
\normalsize
In what follows, we replace $\varphi^l(l_{ik}(s), \boldsymbol{n}_{ik}(s), s)$ by $\varphi^l( \boldsymbol{n}_{ik}(s), s)$ since $\varphi_l$ is independent on $l_{ik}(s)$. Therefore,
\small
\begin{eqnarray}
&&\hspace{-0.5 cm}A_{ik}\varphi(l_{ik}(s), \boldsymbol{n}_{ik}(s), s) =
\sum_{j \in \mathcal{S}}n_{ik}^j G^{ik}_{jl}(\alpha^{N,j}_{ik}(t))(\varphi_l(\boldsymbol{n}_{ik}(s)+\boldsymbol{e}_{lj},h)-\varphi_l(\boldsymbol{n}_{ik}(s),s)) \nonumber\\
&&\hspace{3 cm}+ \sum_{j\neq l}n_{ik}^l G^{ik}_{lj}(\alpha^{N,l}_{ik}(t))(\varphi_l(\boldsymbol{n}_{ik}(s)+\boldsymbol{e}_{jl},s)-\varphi_l(\boldsymbol{n}_{ik}(s),s)),\nonumber\\
&&\hspace{-0.5 cm}=\Big(2(\frac{n^l_{ik}(s)}{N_{ik}}-m^l_{ik}(s))+\frac{1}{N_{ik}}\Big)\sum_{j \neq l}\frac{n^j_{ik}(s)}{N_{ik}}G^{N,ik}_{jl}(\alpha^{N,j}_{ik}(s))-(2(\frac{n^l_{ik}(s)}{N_{ik}}-m^l_{ik}(s))-\frac{1}{N_{ik}})\sum_{j \neq l}\frac{n^l_{ik}(s)}{N_{ik}}G^{N,ik}_{lj}(\alpha^{N,l}_{ik}(s)).\nonumber
\end{eqnarray}
\normalsize
Now, using the property that $\sum_{j \neq l}G^{N,ik}_{lj}(\alpha^{N,l}_{ik}(s))=-G^{N,ik}_{ll}(\alpha^{N,l}_{ik}(s))$, we have
\small
\begin{eqnarray}
\hspace{-0.6 cm}A_{ik}\varphi(l_{ik}(s), \boldsymbol{n}_{ik}(s), s) &=&\Big(2(\frac{n^l_{ik}(s)}{N}-m^l_{ik}(s))+\frac{1}{N}\Big)\sum_{j \neq l}\frac{n^j_{ik}(s)}{N}G^{N,ik}_{lj}(\alpha^{N,l}_{ik}(s))\nonumber
\end{eqnarray}
\begin{eqnarray}
&&+(2(\frac{n^l_{ik}(s)}{N_{ik}}-m^l_{ik}(s))-\frac{1}{N_{ik}})\frac{n^l_{ik}(s)}{N}G^{N,ik}_{ll}(\alpha^{N,l}_{ik}(s)),\nonumber
\end{eqnarray}
\begin{eqnarray}
&&\leq \Big(2(\frac{n^l_{ik}(s)}{N_{ik}}-m^l_{ik}(s)\Big)\sum_{j}\frac{n^j_{ik}(s)}{N}G^{N,ik}_{lj}(\alpha^{N,l}_{ik}(s))+\frac{K_5}{N_{ik}},
\end{eqnarray}
\normalsize
where the last equality follows since each transition rate is bounded. Thus,
\small
\begin{eqnarray}
&&V^N_{ik}(l,t)\leq \mathbb{E} \int_{0}^t 2(\frac{n^l_{ik}(s)}{N_{ik}}-m^l_{ik}(s))\sum_{j}\Big(\frac{n^j_{ik}(s)}{N}G^{ik}_{jl}(\alpha^{j}_{ik}(s))-m^j_{ik}(s)G^{ik}_{jl}(\alpha^{N,j}_{ik}(s))\Big) +\frac{K_5}{N_{ik}},\nonumber
\end{eqnarray}
\begin{eqnarray}
&&=\mathbb{E} \int_{0}^t2(\frac{n^l_{ik}(s)}{N_{ik}}-m^l_{ik}(s))\sum_{j}\frac{n^j_{ik}(s)}{N_{ik}}(G^{ik}_{jl}(\alpha^{N,j}_{ik}(s))-G^{ik}_{jl}(\alpha^j_{ik}(s))\nonumber
\end{eqnarray}
\begin{eqnarray}
&&+G^{ik}_{jl}(\alpha^{j}_{ik}(s))((\frac{n^j_{ik}(s)}{N_{ik}}-m^j_{ik}(s))ds+\frac{K_6}{N_{ik}},\label{vlk}
\end{eqnarray}
\normalsize
where $K_6=K_5 \cdot T$.
Since in our game, the transitional rate is Lipchitz in $m_{ik}(t)$ $\forall$ $(i,k)$ and in $\boldsymbol{u}_{ik}$ (according to Proposition 1), and using Remark \ref{grad} we have for $(i,k)=(i',k')$ ($(i',k')$ is the class of the reference player)
\small
\begin{eqnarray}
&&G^{ik}_{jl}(\alpha^{N,j}_{ik}(s))-G^{ik}_{jl}(\alpha^j_{ik}(s))\nonumber
\end{eqnarray}
\begin{eqnarray}
&&\leq K_7 (\sum_{(r,v) \in \mathcal{C}}||\boldsymbol{m_{rv}}(s)-\frac{\boldsymbol{n_{rv}}(s)+\boldsymbol{e}_{jl}}{N_{rv}}||)+(||\boldsymbol{u}^{N,\boldsymbol{n}+e^{ik}_{jl}}_{ik}(s)-\boldsymbol{u}_{ik}||),\nonumber
\end{eqnarray}
\begin{eqnarray}
&&\leq K_7(\sum_{(r,v) \in \mathcal{C}}||\boldsymbol{m}_{rv}(s)-\frac{\boldsymbol{n_{rv}}(s)}{N_{rv}}||)+\frac{2}{N_{rv}}+(||\boldsymbol{u}^{N,\boldsymbol{n}+e^{ik}_{jl}}_{ik}(s)-\boldsymbol{u}^{N,\boldsymbol{n}}_{ik}(s)||+||\boldsymbol{u}^{N,\boldsymbol{n}}_{ik}(s)-\boldsymbol{u}_{ik}(s)||),\nonumber\\
&&\leq K_7(\sum_{(r,v) \in \mathcal{C}}||\boldsymbol{m}_{rv}(s)-\frac{\boldsymbol{n_{rv}}(s)}{N_{rv}}||)+\frac{2+2K_8}{N_{rv}}+||\boldsymbol{u}^{N,\boldsymbol{n}}_{ik}(s)-\boldsymbol{u}_{ik}(s)||. \label{alphalip}
\end{eqnarray}
\normalsize
Also, for $(i,k) \neq (i',k')$, we have
\small
\begin{eqnarray}
&&G^{ik}_{jl}(\alpha^{N,j}_{ik}(s))-G^{ik}_{jl}(\alpha^j_{ik}(s))\nonumber
\end{eqnarray}
\begin{eqnarray}
&&\leq K_7 (\sum_{(r,v) \in \mathcal{C}}||\boldsymbol{m_{rv}}(s)-\frac{\boldsymbol{n_{rv}}(s)-\boldsymbol{e}_{l}}{N_{rv}}||)+(||\boldsymbol{u}^{N,\boldsymbol{n}+e_{jl}}_{ik}(s)-\boldsymbol{u}_{ik}||),\nonumber
\end{eqnarray}
\begin{eqnarray}
&&\leq K_7(\sum_{(r,v) \in \mathcal{C}}||\boldsymbol{m}_{rv}(s)-\frac{\boldsymbol{n_{rv}}}{N_{rv}}(s)||)+\frac{1}{N_{rv}}+(||\boldsymbol{u}^{N,\boldsymbol{n}-e^{ik}_{l}}_{ik}(s)-\boldsymbol{u}^{N,\boldsymbol{n}}_{ik}(s)||+||\boldsymbol{u}^{N,\boldsymbol{n}}_{ik}(s)-\boldsymbol{u}_{ik}(s)||),\nonumber
\end{eqnarray}
\begin{eqnarray}
&&\leq K_7(\sum_{(r,v) \in \mathcal{C}}|||\boldsymbol{m}_{rv}(s)-\frac{\boldsymbol{n_{rv}}(s)}{N_{rv}}||)+\frac{2+2K_8}{N_{rv}}+||\boldsymbol{u}^{N,\boldsymbol{n}}_{ik}(s)-\boldsymbol{u}_{ik}(s)||. \label{alphalip}
\end{eqnarray}
\normalsize
By substituting (\ref{alphalip}) into (\ref{vlk}), we get
\small
\begin{eqnarray}
&&\hspace{-0.5 cm}V^N_{ik}(l,t)\leq 2K_7 \mathbb{E} \int_{0}^t|\frac{n^l_{ik}(s)}{N_{ik}}-m^l_{ik}(s)|\Big((\sum_{(r,v) \in \mathcal{C}}||\boldsymbol{m}_{rv}(s)-\frac{\boldsymbol{n_{rv}}(s)}{N_{rv}}||)+\frac{K_9}{N_{rv}}+||\boldsymbol{u}^{N,\boldsymbol{n}}_{ik}(s)-\boldsymbol{u}_{ik}(s)||\Big)ds\nonumber
\end{eqnarray}
\begin{eqnarray}
&&+\mathbb{E} \int_{0}^t 2 \Big(\frac{n^l_{ik}(s)}{N_{ik}}-m^l_{ik}(s)\Big)\sum_{j}G^{ik}_{jl}(\alpha_{ik}(t))\Big(\frac{n^j_{ik}(s)}{N_{ik}}-m^j_{ik}(s)\Big)ds +\frac{K_6}{N},\nonumber
\end{eqnarray}
\begin{eqnarray}
&&\leq   2K_7 \mathbb{E} \int_{0}^t|\frac{n^l_{ik}(s)}{N_{ik}}-m^l_{ik}(s)|\Big((\sum_{(r,v) \in \mathcal{C}}||\boldsymbol{m}_{rv}(s)-\frac{\boldsymbol{n_{rv}}(s)}{N_{rv}}||) +\frac{K_{10}}{N_{rv}}+||\boldsymbol{u}^{N,\boldsymbol{n}}_{ik}(s)-\boldsymbol{u}_{ik}(s)||\Big)ds\nonumber,
\end{eqnarray}
\begin{eqnarray}
&&+K_9\mathbb{E} \int_{0}^t 2 \Big|\frac{n^l_{ik}(s)}{N_{ik}}-m^l_{ik}(s)\Big|\sum_{j}\Big|\frac{n^j_{ik}(s)}{N_{ik}}-m^j_{ik}(s)\Big|ds,  \label{sixtythree}
\end{eqnarray}
\vspace{-0.3 cm}
\normalsize
where $K_{10}=K_6+2T(1+K_7)+K_9$.
Let $(y,z)=\text{argmax}_{(r,v)}||\boldsymbol{m}_{rv}(s)-\frac{\boldsymbol{n}_{rv}(s)}{N_{rv}}||$. Thus, 
\vspace{0.2 cm}
\small
\begin{eqnarray}
|\frac{n^l_{ik}(s)}{N_{ik}}-m^l_{ik}(s)|\Big((\sum_{(r,v) \in \mathcal{C}}||\boldsymbol{m}_{rv}(s)-\frac{\boldsymbol{n_{rv}}(s)}{N_{rv}}||)\leq 
\sum_{(r,v) \in \mathcal{C}} ||\boldsymbol{m}_{yz}(s)-\frac{\boldsymbol{n}_{yz}(s)}{N_{yz}}||^2, \nonumber\\
\leq |\mathcal{C}|||\boldsymbol{m}_{yz}(s)-\frac{\boldsymbol{n}_{yz}(s)}{N_{rv}}||^2,
 \label{sixtyfive}
\end{eqnarray}
\normalsize
\vspace{-1.3 cm}
\small
\begin{eqnarray}
|\frac{n^l_{ik}(s)}{N_{ik}}-m^l_{ik}(s)|.||\boldsymbol{u}^{N,\boldsymbol{n}}_{ik}(s)-\boldsymbol{u}_{ik}(s)|| \leq ||\frac{\boldsymbol{n}_{ik}(s)}{N_{ik}}-\boldsymbol{m}_{ik}(s)||.||\boldsymbol{u}^{N,\boldsymbol{n}}_{ik}(s)-\boldsymbol{u}_{ik}(s)||, \nonumber\\
\leq ||\frac{\boldsymbol{n}_{ik}(s)}{N_{ik}}-\boldsymbol{m}_{ik}(s)||^2+||\boldsymbol{u}^{N,\boldsymbol{n}}_{ik}(s)-\boldsymbol{u}_{ik}(s)||^2. \label{sixtysix}
\end{eqnarray}
\normalsize
From (\ref{sixtythree}), (\ref{sixtyfive}), and (\ref{sixtysix}), we have
\small
\begin{eqnarray}
&&V^N_{ik}(l,t) \leq K_{11} \mathbb{E} \int_{0}^t ||\frac{\boldsymbol{n}_{ik}(s)}{N_{ik}}-\boldsymbol{m}_{ik}(s)||^2+||\boldsymbol{u}^{N,\boldsymbol{n}}_{ik}(s)-\boldsymbol{u}_{ik}(s)||^2 + ||\frac{\boldsymbol{n}_{yz}(s)}{N_{yz}}-\boldsymbol{m}_{yz}(s)||^2+\frac{K_{10}}{N_{\max}},\nonumber\\
\end{eqnarray}
\normalsize
where $K_{11}=2K_7+K_9$ and $N_{\max}=\max_{(r,v)}N_{rv}$. Thus,
\small
\begin{equation}
V^N_{ik}(t) \leq C_2 \mathbb{E} \int_{0}^t (V^N_{ik}(s)+W^N_{ik}(s)+V^N_{yz}(s)) ds +\frac{C_2}{N_{\max}},
\end{equation}
\normalsize
with $C_2=\max\{K_{11},K_{10}\}$.

%

\normalsize

\vspace{-0.2 cm}
\def\baselinestretch{0.95}

\end{document}